\documentclass[12pt,fleqn]{article}
\usepackage{amsfonts,ulem,url}
\usepackage{amssymb}
\usepackage{amsmath}
\usepackage{amstext,verbatim,graphicx,ifthen,multirow,amsthm}
\usepackage{bm}
\usepackage{natbib}
\usepackage[bf]{caption}
\usepackage{setspace}
\usepackage{appendix}
\usepackage{float}

\newtheorem{remark}{Remark}

\renewcommand{\mathbf}{\boldsymbol}
\renewcommand{\thepage}{}

\newcommand{\weight}{\gamma}

\newcommand{\uw}{\underline{w}}
\newcommand{\uww}{\underline{W}}
\newcommand{\uyy}{\underline{Y}}

\newcommand{\mmw}{\mathbb{W}}

\newcommand{\indep}{\perp\!\!\!\perp}

\newcommand{\assignment}{{\rm design}}

\newcommand{\double}{{\rm dr}}

\newcommand{\mbw}{\mathbf{W}}
\newcommand{\mby}{\mathbf{Y}}

\newcommand{\mme}{\mathbb{E}}

\newcommand{\ow}{\overline{W}}

\newcommand{\outcome}{{\rm outc}}

\newcommand{\fe}{{\rm fe}}

\newtheorem{assumption}{Assumption}[section]
\newtheorem{theorem}{Theorem}
\newtheorem{lemma}{Lemma}
\newtheorem{prop}{Proposition}

\def\monthname{\ifcase\month\or
  January\or February\or March\or April\or May\or June\or July\or
  August\or September\or October\or November\or December\fi}
\numberwithin{equation}{section}

\DeclareMathOperator*{\argmin}{arg\,min}

\def\monthname{\ifcase\month\or
January\or February\or March\or April\or May\or June\or
July\or August\or September\or October\or November\or December\fi}
\renewcommand{\appendix}{\small\parindent 0cm\setcounter{equation}{0}
\renewcommand{\theequation}{A.\arabic{equation}}
\setcounter{lemma}{0}\renewcommand{\thelemma}{A.\arabic{lemma}}
\setcounter{theorem}{0}\renewcommand{\thetheorem}{A.\arabic{theorem}}}

\usepackage[pdftex,colorlinks=true,unicode,bookmarksnumbered=false, hyperfootnotes=true]{hyperref}
\hypersetup{citecolor = blue}

\graphicspath{{./figures/}}
\usepackage{subfiles}

\begin{document}

\title{\textbf{Double-Robust Identification for Causal Panel Data Models}\thanks{{\small This paper benefited greatly from our discussions with Manuel Arellano, St\'{e}phane Bonhomme, and David Hirshberg. We are grateful for comments from seminar participants at CERGE-EI, University of Chicago, University of Georgia, Princeton University, and various conferences.
This research was generously supported
by ONR grant N00014-17-1-2131. }} }
\author{Dmitry  Arkhangelsky \thanks{{\small  Associate Professor, CEMFI, darkhangel@cemfi.es. }} \and Guido W. Imbens\thanks{{\small Professor of
Economics,
Graduate School of Business and Department of Economics, Stanford University, SIEPR, and NBER,
imbens@stanford.edu.}} }
\date{\ifcase\month\or
January\or February\or March\or April\or May\or June\or
July\or August\or September\or October\or November\or December\fi \ \number%
\year\ \  (First version September 2019)}
\maketitle\thispagestyle{empty}

\begin{abstract}
\singlespacing
\noindent We study identification and estimation of causal effects in settings with panel data. 
Traditionally researchers follow model-based identification strategies relying 
 on assumptions governing the relation between the potential outcomes and the observed and unobserved confounders. We focus on a different, complementary,
approach to identification where assumptions are made about the relation between the treatment assignment and the unobserved confounders. Such strategies are common in cross-section settings but have rarely been used with panel data.
We introduce different sets of assumptions that follow the two paths to identification, and develop a double robust approach. We propose estimation methods that build on these identification strategies.
\end{abstract}

\noindent \textbf{Keywords}: fixed effects, cross-section data, clustering, causal effects, treatment effects, unconfoundedness.

\begin{center}
\end{center}



\baselineskip=20pt\newpage
\setcounter{page}{1}
\renewcommand{\thepage}{\arabic{page}}
\renewcommand{\theequation}{\arabic{section}.\arabic{equation}}

\section{Introduction}
\label{section:intro}
\setcounter{equation}{0}

Panel data are widely used to estimate causal effects of policy interventions on economic outcomes. Such data are particularly useful in settings where researchers are concerned with  the presence of unobserved time-invariant confounders that invalidate simple comparisons between treated and control  outcomes. With $Y_{it}$  denoting the outcome of interest for unit $i$ (for $i=1,\ldots,N$) in period $t$ (for $t=1,\ldots,T$), the general setup we consider is
\begin{equation} 
Y_{it}=g_t(W_{it},U_i,X_{it},\varepsilon_{it}),
\end{equation}
with
 $W_{it}$  an indicator for the treatment, $U_i$ the unobserved confounder, $X_{it}$ the observed attributes/confounders/covariates,
 and $\varepsilon_{it}$ an independent idiosyncratic error term.
The possibility that $U_i$ may be correlated with 
 $W_{it}$ even after controlling for observed confounders prevents us from estimating the average effect of $W_{it}$ on the outcome by comparing covariate-adjusted treated and control outcomes.
  
The fundamental challenge in estimating treatment effects with panel data is that we wish to adjust flexibly both for differences between units and differences between periods. Making each of these adjustments separately is simple. For example, we can control for differences between units by comparing outcomes for treated and control periods for the same unit. Similarly, we can eliminate the differences between time periods by comparing treated and control outcomes in the same period. However, we cannot make both of these adjustments simultaneously because there is no variation in the treatment for a given unit $i$ in a given period $t$. As a result, there has to be some compromise, and we need to resort to comparisons that adjust only partially for differences between units and periods. 

Often such compromise is reached by imposing functional form restrictions on the outcome model $g_t(\cdot)$ that allow us to remove  dependence on $U_i$. A common approach in empirical work relies on a linear additive two-way fixed effect specification for the outcome model,
\begin{equation}\label{eq:fe_specif}
g_t(w,u,x,e)=\alpha(u)+\lambda_t+w\tau+x^\top\beta+e,
\end{equation}
in combination with $\varepsilon_{it}\indep \{(W_{il},X_{il})\}_{l=1}^T$,
so that the parameters, including the coefficient of interest $\tau$, can be estimated by least squares regression. Although widely used in practice, there  are concerns that the model relies on unreasonably strong assumptions. 

The first contribution of the paper is our focus on a different,  design-based, strategy. The basic omitted-variable-bias insight from the cross-section literature (going back to at least the 1960s, see \cite{goldberger1991course, angristpischke, wooldridge2010econometric} for  textbook discussions) implies that the bias from an unobserved confounder comes from 
the combination of its correlation with the outcome and its correlation with $W_{it}$. As an alternative, or complement, to building a model for the outcomes that restricts the dependence of the outcomes on the unobserved confounders, we can  therefore  build a model for the assignment mechanism to remove the bias. 
Let $\uww_i$ be the $T$-vector of assignments with typical element $W_{it}$. The modeling restrictions we consider are of the form
\begin{equation}\label{twee}
\uww_i \ \indep\  U_i\   \Bigl|\ S_i,
\end{equation}
where $S_i$ is a known function of $\uww_i$.  For example, $S_i$ may be equal to the average treatment assignment for unit $i$, $\ow_i=\sum_t W_{it}/T$. This would amount to the assumption that units with the same fraction of treated periods are comparable. The general strategy of removing the endogeneity by conditioning on a cluster-specific statistic goes back to  \cite{altonji2005cross}. Later in the paper we provide several examples of statistical and structural economic models that justify (\ref{twee}) for a particular choice of $S_i$.
If (\ref{twee}) holds, we can compare treated and control units as long as they have the same value for $S_i$.  This approach provides a middle ground between a nonparametric within-unit analysis of \cite{chernozhukov2013average} and standard cross-sectional techniques. In applications we expect $S_i$  to capture a substantial part of the unobserved heterogeneity even if (\ref{twee}) does not formally hold, making the doubly robust strategy we discuss next appealing.

 The second contribution of the paper is the insight that one can combine the two strategies, model-based and design-based, into a single, doubly robust identification strategy. Models such as (\ref{eq:fe_specif}) validate a particular set of comparisons between treated and control outcomes. The restriction in (\ref{twee}) validates a different set of comparisons between treated and control outcomes. In many cases we can focus on comparisons that are validated by both the model in (\ref{eq:fe_specif}) and the restriction in (\ref{twee}). Such comparisons remain meaningful as long as at least one of the models is correct. Note that the double robustness here refers to the identification, rather than to the estimation given a common set of assumptions as in much of the literature on double robustness (\citet{robins1,chernozhukov2018double}).

To implement our strategy we restrict attention to linear estimators that are widely used in economics and statistics (e.g., \cite{donoho1994statistical,armstrong2018optimal}). Our estimator has the following form:\begin{equation}    \hat \tau = \frac{1}{NT}\sum_{i,t}\gamma_{it}Y_{it}\end{equation}where researchers explicitly select the weights $\gamma_{it}$ by solving a quadratic optimization problem. The weights depend on the treatment indicators and possibly on covariates, but not on the outcome data. Our estimator remains consistent if either outcome or assignment model is correctly specified and is robust to arbitrary heterogeneity in treatment effects. We also provide an extension to general non-binary treatments.

Our strategy is not based on constructing consistent estimates for $U_i$ and then controlling for the estimated $\hat U_i$. In fact, in the  fixed $T$ case  we consider, it is  impossible to do that.  Instead, we leverage the fact that the distribution of $U_i$ is the same for units with different assignment paths $\uww_i$ as long as we restrict our attention to subpopulations defined by $S_i$. This emphasizes the practical role that design assumptions can play in models with unobserved heterogeneity.
We can re-interpret the conventional two-way fixed effect estimator in this framework. Specifically, it can be viewed as comparing treated and control units at the same point in time, within the set of units with the same fraction of treated periods, that is, conditioning on $S_i=\sum_{t=1}^T W_{it}/T$ (e.g., \cite{mundlak1978pooling}). However, as we show by a simple example, the two-way fixed effect estimator is not doubly robust, and our proposed estimator explicitly addresses this problem.

\section{Literature}
\setcounter{equation}{0}

The two-way additive fixed effect structure in  (\ref{eq:fe_specif})  has a long history in applied economics (going back at least to \cite{mundlak1961empirical, hoch1962estimation, mundlak1965consequences}) and econometrics (e.g., \citet{chamberlain1984panel,chamberlain1992efficiency, arellano2011identifying,graham2012identification,chernozhukov2013average,freyberger2018non}). The original motivation comes from an underlying structural model that describes the link between unobserved factors and outcomes. The relationship between outcomes and assignments is often left unrestricted, except for high-level exogeniety conditions (see \cite{arellano2003panel} for a discussion). 

Restrictions like (\ref{twee}) are different in spirit and can be justified by making additional assumptions on the relationship between $\uww_i$ and $U_i$, or, in other words, by formulating an assignment model. This type of restrictions on the assignment mechanism are at the heart of the ``credibility revolution'' in applied economics (\cite{angrist2010credibility}) that emphasizes the role of the research design. Many strategies  currently used in empirical work for causal inference with cross-section data are based on such design assumptions (see \citet{angristpischke,currie2020technology} for evidence on this). Although  less common than outcome modeling in panel data settings, this design-based approach has been used to achieve identification in settings with grouped data, e.g., the exchangeability assumption in \cite{altonji2005cross}, or the exponential family assumption in \cite{arkhangelsky2018role} (see also \cite{borusyak2020non}). 

We are making two contributions to this literature. First, building on the research on binary panel models (e.g., \cite{honore2000panel,chamberlain2010binary,aguirregabiria2018sufficient}), we show how to use design-based assumptions to identify treatment effects in this setting. Second, we propose a doubly robust identification strategy that combines the models for outcomes and assignments and remains valid if either of them is correctly specified. In practice, this means that applied researchers can directly exploit information about economic mechanisms behind different patterns in $\uww_i$, without abandoning familiar outcome models such as (\ref{eq:fe_specif}).

This paper is also directly related to recent causal panel literature that focuses on two-way fixed effect estimators (\cite{de2018two, callaway2018difference, sant2020doubly, goodman2021difference, athey2022design}). Similarly to these papers, we use the two-way structure to model the baseline outcomes. However, our focus is quite different. First,  we consider general designs, whereas the previous research has been focused on either block case or staggered adoption. Second, we develop a new estimator, that directly exploits both an assignment and an outcome model. Finally,  our identification strategy can be applied more broadly. In particular, 
although we focus on the two-way model for the baseline outcomes, the strategy
 can be extended to general factor models, thus connecting to the literature on synthetic control (e.g.,
\cite{Abadie2010, xu2017generalized, athey2017matrix, Abadie2010, ben2018augmented, arkhangelsky2021synthetic, chernozhukov2019inference}).

Our analysis focuses on static models, assuming away any dynamic effects of past treatments on contemporaneous outcomes. This is not to say that dynamic effects are not common or important in empirical practice. They are, and there are important conceptual issues in that setting that are not present in the current one we consider. Still, we focus on the static case for three reasons. First, the presence of dynamic effects does not eliminate the issues we raise in the current paper, and we feel that they are exposited most clearly in this setting. Second, there are many applications where the dynamic effects are second order. For example, for analyzing the impact of promotions on purchases of non-durables, the dynamic effects are likely to be modest. More generally, this holds when outcomes measure activity over a short period of time  (e.g., days or weeks) and are measured far apart in time (e.g., years). 
Finally, the main challenge in allowing for dynamic effects is the need for a model for the dynamic component. Current applied work does not provide a commonly accepted benchmark for such a model, making a doubly robust strategy infeasible. Nevertheless, our design-based approach remains valid regardless of the presence or absence of dynamic effects and can be directly used as input for propensity score estimators introduced in \cite{bojinov2021panel}.

\section{Setup}\label{sec:setup}
\setcounter{equation}{0}
We have access to a panel data set with $N$ units observed over $T$ periods ($i$ and $t$ being a generic unit and period, respectively). We focus on asymptotic approximations based on large $N$ and fixed $T$. We are interested in the effect of a binary policy variable $w$ on some economic outcome $Y_{it}$. Later we discuss settings with more general treatments or policies. To formalize this we consider a potential outcome framework (\citet{rubin1974estimating, imbens2015causal}). The policy or treatment  can change over time, and so is indexed by unit $i$ and time $t$, $W_{it}\in\{0,1\}$. Let $\uw^t\equiv (w_1,w_{2},\ldots,w_t)$ denote the sequence of treatment exposures up to time $t$,  with $\uw$ as shorthand for the full vector of exposures $\uw^T$.  
Define $\uww_{i} \equiv (W_{i1},\dots, W_{iT})$  to be the full  assignment vector for unit $i$.  For the first part of the paper we abstract from the presence of additional unit-level covariates $X_i$. We explicitly account for their presence in Section \ref{sec:est_inf}.

Let $Y_{it}(\uw^t)$ denote the potential outcome for unit $i$ at time $t$, given treatment history up to time $t$:
\begin{equation}
Y_{it}(\uw^t)\equiv Y_{it}(w_1,w_{2},\dots, w_{t}).
\end{equation}
 In this paper we consider a static version of this general model and make the following assumption.
\begin{assumption}{\sc(No Dynamics)}\label{as:static}
For arbitrary $t$-compponent assignment vectors $\uw$ and $\uw^{\prime}$ such that the period $t$ assignment is the same, $w_{t} = w^{\prime}_{t}$ the potential outcomes in period $t$ are the same:
\begin{equation}\label{st_mod}
Y_{it}(\uw) = Y_{it}(\uw^{\prime}).
\end{equation}
\end{assumption}
This assumption is motivated by two observations. First,  in many economic applications dynamic effects are modest relative to the contemporaneous effects, making the static model a reasonable approximation and a common choice in empirical work. Second, while our insights are relevant for the general case, the conceptual issues are more easily exposited in this simpler setting.
Crucially, Assumption \ref{as:static} does not restrict heterogeneity in contemporaneous treatment effects in any way.

Given the  no-dynamics assumption we can index the potential outcomes by a single binary argument $w$, so we write $Y_{it}(w)$, for $w\in\{0,1\}$. 
Define 
\begin{equation}
    \uyy_i(\uw) \equiv (Y_{i1}(w_1),\dots, Y_{iT}(w_T))
\end{equation}
to be the $T$-component vector of potential outcomes.  In this setup we are interested in various treatment effects. The main building block for such average treatment effects is the  individual and time-specific treatment effect:
\begin{equation}
\tau_{it}\equiv Y_{it}(1) - Y_{it}(0)
\end{equation}
We focus primarily on identification and estimation of average treatment effects, typically a convex combination of individual effects $\tau_{it}$.

We 
make the following assumption:
\begin{assumption}{\sc(Latent Unconfoundedness)}\label{as:unc}
There exist a random variable $U_i\in \mathbb{R}^d$ such that the following conditional independence holds:
\begin{equation}
\uww_i  \indep \ \Bigl\{\uyy_{i}(\uw)\Bigr\}_{\uw}\ \Bigl|\ U_i
\end{equation}
\end{assumption}
We do not restrict $U_i$ to be scalar, nor do we restrict the functional form of the relation between $\underline{W}_i$ and $U_i$, or the relation between $Y_i(\underline{w})$ and $U_i$.

This assumption implies that once we control for the (unobserved, and potentially vector-valued) $U_i$ all the differences in the treatment paths $\uww_i$ across units are unrelated to the potential outcomes. It is natural when $W_{it}$ is either driven by some (quasi)-experimental shocks or is a choice variable that is not used to optimize $Y_{it}$ directly. Examples of the first type are common in the applied literature, especially when the shocks are aggregate, but some units are more exposed to them than others (see Section \ref{sec:emp_il}). As an example of the second type, consider a situation where $W_{it}$ corresponds to national-level prices for product $i$ in period $t$ and $Y_{it}$ is a measure of sales in a particular local market. Unobserved quality $U_i$ makes $Y_{it}$ and $W_{it}$ correlated (over $i$), despite the fact that $W_{it}$ is not chosen to optimize $Y_{it}$ directly. Overall, the relationship between $W_{it}$ and $U_i$ that satisfies Assumptions \ref{as:unc} can arise for a variety of reasons and in the subsequent sections we show why it is useful to model it explicitly.

\section{Two Paths to Identification: An Example }\label{sec:bas_id}
\setcounter{equation}{0}

\subsection{Preliminaries}
Let $\mbw$ be the support of the vector of assignments $\uww_i$; we can think of $\mbw$ as a matrix with at most $2^T$ rows (because $W_{it}$ is binary)  and $T$ columns, where each row is an element of the support of $\uww_i$. Let $\mathbf{W}_k$ be a $k$ row of the matrix $\mbw$ -- a $T$-dimensional vector of zeros and ones. Let $\pi_k \equiv{\rm pr}(\uww_i=\mathbf{W}_k)= \mme\left[\mathbf{1}_{\uww_i = \mbw_k}\right]$. All $\pi_k$ are positive, otherwise the corresponding row of $\mbw$ can be dropped. Let $K\leq 2^T$ be the number of rows in $\mbw$.

For example, if $T=3$ then a possible form for $\mathbf{W}$ is:
\begin{equation}\label{example_dist}
\mathbf{W} = \begin{pmatrix}
0 & 0 & 0\\
1 & 0 & 1\\
0 & 1 & 1\\
1 & 1 & 1\\
\end{pmatrix}
\end{equation}
Each row of this matrix represents a possible assignment, and in this particular case only four  out of the eight $(2^3 = 8)$ possible combinations have positive probability. For a particular unit $i$, let $k(i)$ be the row $\mbw_k$ of the support matrix $\mbw$ such that $\mbw_{k(i)}=\uww_i$. For the identification argument we initially assume we know $\mbw$ and the assignment probabilities $\pi_k$. We consider the case with unknown $\pi$  in Section \ref{sec:est_inf}.

We are interested in estimating weighted averages of the treatment effects $\tau_{it}$. Following Assumption \ref{as:static} we have the following representation for the observed outcomes $Y_{it}$:
\begin{equation}
    Y_{it} = W_{it}Y_{it}(1) + (1-W_{it})Y_{it}(0),
\end{equation}
which we collect into $N\times T$ matrix $\mby$. Our estimators will be linear in $\mby$, with weights $\weight_{it}$:
\[ \hat\tau(\weight)=\frac{1}{NT}\sum_{i=1}^N\sum_{t=1}^T \weight_{it} Y_{it}.\]
For the class of estimators we consider the weights $\weight_{it}$ are a function of $\uww_i$ so that $\weight_{it} = \weight_t(\uww_i)$, satisfying the restrictions $\frac{1}{NT}\sum_{i,t} W_{it}\gamma_{it}=1$ and $\frac{1}{NT}\sum_{i,t} (1-W_{it})\gamma_{it}=-1$. Choosing an estimator  corresponds to choosing weight functions $\weight_t(\cdot)$ and vice versa. 

 \subsection{Doubly Robust Identification -- An Example}
  \begin{table}[t]
 \caption{Assignment process and weights}
 \label{table:exmp_1}
\centering
\begin{tabular}{rrrrrr}
  \hline
$k$ & $\mbw_k$ & $\pi_k$& $\weight^{(fe)}_{k1}$ & $\weight^{(fe)}_{k2}$ & $\weight^{(fe)}_{k3}$\\ 
  \hline
1& (0,0,0)&0.09 &  0.46 & -0.64 & 0.18 \\ 
2&  (1,0,0) & 0.04&  5.70 & -3.26 & -2.44 \\ 
3 & (0,1,0)&  0.11 &    -2.16 & 4.60 & -2.44 \\ 
4&  (1,1,0) & 0.14 &   3.08 & 1.98 & -5.07 \\ 
 5& (0,0,1)& 0.07 &   -2.16 & -3.26 & 5.42 \\ 
6& (1,0,1)&  0.08 &  3.08 & -5.88 & 2.80 \\ 
7& (0,1,1)&  0.15 &    -4.78 & 1.98 & 2.80 \\ 
8&  (1,1,1)& 0.32 &  0.46 & -0.64 & 0.18 \\ 
   \hline
\end{tabular}
\end{table}

We start with an example that illustrates the main message of the paper. Suppose that $T = 3$ and $K = 8 = 2^T$, so $\uww_{i}$ has full support. We assume that the distribution of $\uww_{i}$ in population is given by the $K$ probabilities $\pi_k$ in the third column of Table \ref{table:exmp_1}. Suppose that in fact the  potential outcomes $Y_{it}(w) $ satisfy a two-way-fixed-effect  structure and that the treatment effect is constant across time and units:
 \begin{equation}\label{eq:basic_fe_model}
\begin{aligned}
&Y_{it}(w) = \alpha(U_i) + \lambda_t + \tau w+ \varepsilon_{it},\\
&\mathbb{E}[\varepsilon_{it}|\uww_{i}, U_i] = 0.
\end{aligned}
\end{equation}
The fixed effect estimators uses least squares with two-way fixed effects to ``estimate'' $\tau$ in population. This procedure leads in large samples to a particular set of weights $\weight^{(\fe)}_{t}(\uww_i)$, and then to the following fixed effect estimand:
\begin{equation}
\tau^{\fe} =\mathbb{E}\left[\frac{1}{T}\sum_{t}Y_{it} \weight^{(\fe)}_{t}(\uww_i)\right].
\end{equation}
The expectation here is taken over the $U_i$, $\{\varepsilon_{it}\}_{t=1}^T$ and the assignment path $\uww_i$. If (\ref{eq:basic_fe_model}) is correctly specified then $\tau^{\fe}=\tau$ -- the common treatment effect -- but in general this equality will not hold.
For the distribution given above the weights implied by the fixed effect estimator are presented in the last three column of Table \ref{table:exmp_1}. 

By construction these weights sum up to $0$ for every row and every column. 
But there are many weights that satisfy these restrictions. The fixed effect specification (\ref{eq:basic_fe_model}) implies that all of them lead to the same estimand, and it is therefore overidentified. The weights $ \weight^{(fe)}_t(\cdot)$ are selected for efficiency reasons, because under homoskedasiticity they lead to an estimator with the least possible variance. In practice, we can have concerns besides efficiency. In particular, we may be worried that the model  (\ref{eq:basic_fe_model}) is misspecified.

  \begin{table}[t]
 \caption{Aggregated weights}
 \label{table:exmp_2}
\centering
\begin{tabular}{rrrr}
  \hline
$\overline{W}_i$ &$\mathbb{E}[\weight^{(\fe)}_1(\uww_i)|\overline W_{i}] $&$\mathbb{E}[\weight^{(\fe)}_2(\uww_i)|\overline W_{i}] $& $\mathbb{E}[\weight^{(\fe)}_3(\uww_i)|\overline W_{i}] $ \\ 
  \hline
0 & 0.46 & -0.64 & 0.18 \\ 
1/3 & -0.73 & 0.60 & 0.13 \\ 
  2/3 & -0.08 & 0.36 & -0.28 \\ 
  1& 0.46 & -0.64 & 0.18 \\ 
   \hline
\end{tabular}
\end{table}

To illustrate this, suppose that DGP for the assignment mechanism $\uww_i$ has the following form (which is consistent with the probabilities in Table \ref{table:exmp_1}):
\begin{equation}\label{eq:basic_des_model}
\begin{aligned}
\text{$\forall (t,t^{\prime}$): } &W_{it} \indep W_{it^{\prime}} | U_i,\hskip1cm
&\mathbb{E}[W_{it}|U_i] = \frac{\exp(\alpha(U_i) + \lambda_t)}{1+\exp(\alpha(U_i) + \lambda_t)}.
\end{aligned}
\end{equation}
As we show in the following section, this pair of restrictions implies a conditional independence restriction:
\begin{equation}\label{eq:basic_design}
\uww_i  \indep \ \Bigl\{\uyy_{i}(\uw)\Bigr\}_{\uw}\ \Bigl|\ \overline W_i
\end{equation}
where $\overline{W}_i\equiv\sum_{t=1}^T W_{it}/T$ is the fraction of treated periods for unit $i$. This leads to a key insight. Suppose the two-way fixed effect outcome model (\ref{eq:basic_fe_model}) is misspecified, but the assignment mechanism (\ref{eq:basic_des_model}) is correctly specified. Also suppose that the treatment effect is constant, $\tau_{it} = \tau$. Then the  estimand $\tau^\fe$ may still be equal to the treatment effect, $(\tau^\fe = \tau)$, as long as the following condition on the weights is satisfied for all $t$ and $\overline W_{i}$:
\begin{equation}
\mathbb{E}[\weight^\fe_t(\uww_i)|\overline W_{i}] = 0 
\end{equation}
This restriction requires the  weights $\weight^{(\fe)}_t(\uww_i)$ to balance out any time-specific function of $\overline W_{i}$. If this restriction is not satisfied, then the differences in average outcomes for treated and control units may partly be attributed to the differences in the baseline outcomes $Y_{it}(0)$.

Table \ref{table:exmp_2} shows that this condition does not hold for the fixed effect weights from Table \ref{table:exmp_1}. 
In other words, although fixed effects weights  balance individual and time effects overall, they  do not necessarily do so within  subpopulations defined by $\overline{W}_i$. This might come as a surprise, because the two-way estimator can be interpreted as controlling for $\overline W_i$. As shown in \cite{mundlak1978pooling}, the fixed effects estimator $\hat\tau^{\fe}$ is numerically equivalent to the estimator based on the following linear regression:
\begin{equation}
    Y_{it} = \alpha + \lambda_t +\tau W_{it} + \eta \overline{W}_i + \tilde \epsilon_{it}
\end{equation}
As a result, when estimating $\tau^{\fe}$ we control for $\overline{W}_i$ only linearly and this is not enough to enforce the full balancing property.

As an alternative to the fixed effect weights one can  use the assignment model (\ref{eq:basic_design}) to construct inverse propensity (IP) weights (e.g., \cite{rosenbaum1983central}) that deliver the  treatment effect $\tau$ if the design model (\ref{eq:basic_design}) is correctly specified. However, such weights would not identify $\tau$ under the outcome model (\ref{eq:basic_fe_model}). To see this, observe that in general we have the following:
\begin{equation}
   \frac{1}{T}\sum_{t=1}^T \gamma^{(IP)}_{t}(\uww_i) =   \frac{1}{T}\sum_{t=1}^T\left(\frac{W_{it}}{\mathbb{E}[W_{it}|\overline W_i]} -  \frac{1-W_{it}}{\mathbb{E}[(1-W_{it})|\overline W_i]}\right) \ne 0,
\end{equation}
and thus IP weights do not balance individual fixed effects in (\ref{eq:basic_fe_model}).

The question now arises whether we can find the weights that deliver $\tau$ under  weaker conditions. In particular, we want the weights to deliver $\tau$ if either the fixed effect model (\ref{eq:basic_fe_model}) or the design process (\ref{eq:basic_design}) is correctly specified, but not necessarily both.  The answer is positive and is given in Table \ref{table:dr_weights}.
\begin{table}[t]
 \caption{Doubly robust weights}
 \label{table:dr_weights}
\centering
\begin{tabular}{rrrr}
  \hline
$(W_1,W_2,W_3)$&$\weight^{(dr)}_1(\uww_k)$ & $\weight^{(dr)}_2(\uww_k)$ & $\weight^{(dr)}_3(\uww_k)$\\ 
  \hline
(0,0,0)&0.00 & 0.00 & 0.00 \\ 
  (1,0,0)&6.59 & -3.95 & -2.64 \\ 
  (0,1,0)&-1.46 & 4.10 & -2.64 \\ 
 (1,1,0)& 3.24 & 1.66 & -4.90 \\ 
  (0,0,1)&-1.46 & -3.95 & 5.42 \\ 
 (1,0,1)& 3.24 & -6.39 & 3.15 \\ 
  (0,1,1)&-4.81 & 1.66 & 3.15 \\ 
 (1,1,1)& 0.00 & 0.00 & 0.00 \\ 
   \hline
\end{tabular}
\end{table}
It is evident that the weights some up to zero for each row and a simple calculation shows that $\mathbb{E}[\weight^{(dr)}_t(\uww_i)|\overline W_{i}] = 0$ for every $t$ and $\overline W_i$. As a result, there is no trade-off in terms of identification and we can construct the estimator that works under both models. There is some cost in terms of precision. If the fixed effect model is correct, the fixed effect estimator is more efficient in general.

So far we have assumed that the treatment effects are constant. It is well-documented that the two-way estimators have problems in cases with heterogenous treatment effects (e.g., see \cite{de2018two}). This is evident from looking at Table \ref{table:exmp_1}: in the last row we assign a negative weight to all treated units in the second period. To guarantee that this does not happen the following restriction can be enforced when the weights are constructed:
\begin{equation}
W_{it}\weight_t(\uww_i) \ge 0
\end{equation}
As we show in the next sections, this does not make the problem much harder computationally and the robust weights from Table \ref{table:dr_weights} are constructed with this restriction in mind.

\section{Two Paths to Identification: The General Case }\label{sec:full_id}
\setcounter{equation}{0}

In this section we present the main identification results.

\subsection{Identification Through the Outcome Model}

We focus on an outcome model with the following structure:
\begin{assumption}\label{as:chamberlain}
The potential outcomes satisfy:
\begin{equation}\label{chamb_model}
\mme[Y_{it}(w)|U_i] = \alpha(U_i) +\lambda_t+ \tau_t(U_i)w.
\end{equation}
\end{assumption}

Versions of this model with restricted heterogeneity in treatment effects have a long tradition in economic literature (starting at least from \cite{mundlak1961empirical}), and is used as a building block in much of the recent causal panel data literature (e.g., \cite{de2018two,callaway2018difference,abraham2018estimating,goodman2021difference,borusyak2016revisiting}). Below we provide a parsimonious characterization of the set of identified (weighted) average effects,  generalizing similar results established in the literature. In the next section we use this characterization to construct a robust estimator.

To identify a convex combination of $\tau_t(U_i)$  we consider the weights $\weight_{kt}$ that satisfy the following four restrictions:
\begin{equation}\label{eq:out_1}
 \begin{aligned}
&\frac{1}{T} \sum_{k=1}^K \sum_{t=1}^T \pi_k\weight_{kt} \mbw_{kt} = 1, \quad
 \forall k, \sum_{t}\weight_{kt}  = 0, \quad\\
&\forall t, \sum_{k=1}^{K} \pi_k \weight_{kt} = 0, \quad
\forall (t, k), \,\weight_{kt}\mbw_{kt} \ge 0.
 \end{aligned}
\end{equation}
These constraints are natural given the outcome model described above. The first and the last restriction insure that we focus on a convex combination of treatment effects. The second and the third restriction guarantee that weights balance out the systematic variation in the baseline outcomes $Y_{it}(0)$. We do not include the analogue of non-negativity constraint for control units, thus allowing for some extrapolation. Depending on application, one might want to impose such constraint as well.

Let $\mmw_\outcome$ be the set of weights $\{\weight_{kt}\}_{k,t}$ that satisfy these restrictions.
For any generic element $\weight \in \mmw_\outcome$ define  $\weight_{t}(\uww_i,\weight)$ to pick out the period $t$ weight for a unit with assignment path $\uww_i$:
\begin{equation}
\begin{aligned}
&\weight_{t}(\uww_i,\weight) \equiv \sum_{k=1}^{K} \weight_{kt} \mathbf{1}_{\uww_{i} = \mbw_{k}}\\
\end{aligned}
\end{equation}
Using these weights we define the following estimand:
\begin{equation}
\tau(\weight) = \mme\left[\frac{1}{T}\sum_{t=1}^TY_{it} \weight_{t}(\uww_i,\weight) \right]
\end{equation}

\begin{prop}\label{pr:id_model}
Suppose Assumptions \ref{as:static}, \ref{as:unc}, and \ref{as:chamberlain} hold, and that $\weight\in\mmw_\outcome$. Then
 $\tau(\weight)$ is a convex combination of $\tau_t(U_i)$ (over $i$ and $t$). 
\end{prop}

\begin{remark}
A certain convex combination of $\tau_t(U_i)$ can be identified whenever $\mathbb{W}_\outcome$ is non-empty. A necessary and sufficient condition for this is simple: the matrix $\mbw$ should contain at least one of the following two submatrices (up to permutation of columns):
\begin{equation}\label{eq:mat_model}
\begin{aligned}
\mbw_1= \begin{pmatrix} 0 & 1\\ 0 & 0\end{pmatrix},\ \ 
\mbw_2= \begin{pmatrix} 0 & 1\\ 1 & 0\end{pmatrix}.
\end{aligned}
\end{equation}
Both situations require presence of adopters (or ``movers'') of different types. The first one corresponds to the familiar difference-in-difference design: there are adopters $i$ of the treatment, and periods $t$ and $t'$  with $(W_{it}=0,W_{it'}=1)$ and in the same periods $t$ and $t'$ non-adopters $i'$ with $(W_{i't}=0,W_{i't'}=0)$. 
In the second case, there are adopters with $(W_{it}=0,W_{it'}=1)$ and units who switch out with $(W_{i't}=1,W_{i't'}=0)$. To put this discussion in perspective, it is not sufficient to have assignment matrices of the type
\[ \mbw_3=\begin{pmatrix} 0 & 0 \\ 1 & 1 \end{pmatrix},\ \ \mbw_4=\begin{pmatrix} 0 & 1 \\ 0 & 1 \end{pmatrix},\ \ \mbw_5= \begin{pmatrix} 1 & 1\\ 1 & 0\end{pmatrix},\]
where with the first design some units are always in the control group and all others are always in the treatment group, in the second design all units adopt the treatment at exactly the same time. The third design is more complicated, because at a first sight $\mbw_5$ looks very similar to $\mbw_1$; the key difference is that with $\mbw_1$ we have a control period and this allows us to deal with unobserved unit-specific differences. With $\mbw_5$ this is no longer feasible and we need to use negative weights that are disallowed. Standard two-way fixed effect  estimator treats $\mbw_1$ and $\mbw_5$ symmetrically and this is the reason why the standard estimand might be outside of the convex hull of treatment effects.  
\end{remark}

\subsection{Identification Through Design}
In this section we consider assignment processes that satisfy a certain sufficiency property. Here we state it  as a high-level assumption, and provide examples of economic models that satisfy this assumption in the next section.

For some  random variable $S_i \equiv S(\uww_i)$ let $r(\uw,s)$ be the generalized propensity score (\cite{imbens2000}):
\begin{equation}
 r(\uw,s)\equiv{\rm pr}(\uww_i=\uw|S_i=s).
\end{equation}
\begin{assumption}\label{as:suf}{\sc(Sufficiency)}
There exist a known $\uww_i$-measurable sufficient statistic $S_i\in \mathbb{S}$ and a subset $\mathbb{A}\subset\mathbb{S}$ such that:
\begin{equation}\label{eq32een}
\uww_i \indep U_i\ \Bigl|\ S_i=s,
\end{equation}
and for all  $s\in \mathbb{A}$:
\begin{equation}
\max_{\uw}\{ r(\uw,s)\} < 1.
\end{equation}
\end{assumption}
This assumption might look restrictive, but a known $S_i$ that satisfies the conditional independence restriction (\ref{eq32een}) always exists. The obvious choice is $S_i=\uww_i$ that satisfies it by construction. Of course, in that case the overlap condition cannot not be satisfied. For this assumption to hold we need to find  different  values for $\uww_i$ that generate the same distribution of $U_i$, but still exhibit variation in some $W_{it}$. For example, in addition to conditioning on the fraction of treated periods, $\overline{W}_i$, we may want to condition on the number of transitions, $\sum_{t=1}^{T-1} W_{it}(1-W_{it+1})$. In the next section we describe examples that show that sufficient statistics $S_i$ that satisfy this condition arise naturally in different empirical settings.

The main implication of the Assumption \ref{as:suf} is summarized by the following weak unconfoundedness proposition (\cite{imbens2000}):
\begin{prop}\label{pr:op_unc}{\sc (Weak Unconfoundedness)}
Suppose Assumptions \ref{as:static}, \ref{as:unc}, and  \ref{as:suf} hold. Then for any $\uw$:
\begin{equation}\label{ass_S}
\mathbf{1}_{\uww_i=\uw}  \indep \ \uyy_{i}(\uw)\ \Bigl|\ S_i.
\end{equation}
\end{prop}
This proposition demonstrates that unconfoundedness conditional on $U_i$ (\ref{twee}) can be transformed into undonfoundedness conditional on $S_i$ (\ref{ass_S}) under the additional assumption that restricts the assignment process.

Let $S_i$ be a potential sufficient statistic. 
Let $\mbw^s$ be a matrix representation of the support of $\uww_i$ conditional on $S_i = s$ and $\mbw^s_{k}$ be a generic row (element of the support). 
For example, if $S_i = \sum_{t} W_{it}/T$ and $\mbw$ is given by (\ref{example_dist}) then $S_i$ takes $3$ possible values and we have the following:
\begin{equation}
\mathbf{W}^{0} =\begin{pmatrix} 
0 & 0 & 0
\end{pmatrix},
\mathbf{W}^{2/3} = \begin{pmatrix} 
1 & 0 & 1\\
0 & 1 & 1\\
\end{pmatrix},
\mathbf{W}^{1} = \begin{pmatrix} 
1 & 1 & 1
\end{pmatrix}
\end{equation}

When considering identification strategy based on design assumptions we do not impose any additional restrictions on potential outcomes besides Assumptions \ref{as:static} and \ref{as:unc}. In this case, one can identify a convex combination of individual treatment effects using the weights that satisfy the following restrictions (for all $k,s$ and $t$):
\begin{equation}\label{design}
\frac{1}{T} \sum_{tk} \pi_{k}\weight_{kt} \mathbf{W}_{kt} = 1, \quad
\sum_{k: \mbw_k \in \mbw^s}\pi_k\weight_{kt} = 0,\quad
\sum_{k: \mbw_k \in \mathbf{W}^s}\pi_k\weight_{kt} \mbw_{kt} \ge 0,
\end{equation}
Let $\mmw_\assignment$ be the set of weights $\{\weight_{tk}\}_{t,k}$ that satisfy these restrictions. It is easy to see that $\mmw_\assignment$ is nonempty whenever there exists at least one value $s$ such that $\mathbf{W}^s$ contains at least two rows. This is guaranteed by the second part of Assumption \ref{as:suf}.   For any $\weight\in \mmw_\assignment$ define the $\weight_{t}(\uww_i,\weight)$ in the same way as before and consider the following expectation:
\begin{equation}
\tau(\weight) = \mathbb{E}\left[\frac{1}{T}\sum_{t=1}^TY_{it} \weight_{t}(\uww_i,\gamma) \right]
\end{equation}

\begin{prop} \label{pr:id_design}
Suppose Assumptions \ref{as:static}, \ref{as:unc}, and \ref{as:suf} hold, and that $\weight\in\mmw_\assignment$. Then $\tau(\weight)$ is a convex combination of treatment effects. 
\end{prop}

\subsection{Examples of Assignment Models}\label{sec:examples}

In this section we consider various examples of assignment models. We show that Assumption \ref{as:suf} holds in a wide range of examples and discuss their applicability to various applied problems.

\subsubsection{Aggregate shocks}\label{sec:ag_shocks}

As a first case, we consider a situation, where $W_{it}$ is determined by an observed $p$-dimensional aggregate shock $\psi_t$, and idiosyncratic noise $\nu_{it}$. Units are affected by $\psi_t$ in a way that is determined by their unobserved $U_i$. Formally, we have the following model that includes a latent index:
\begin{equation}
\begin{aligned}
&W^{\star}_{it} = \theta_t+ \alpha_1(U_i)+\alpha_2^T(U_i)\psi_t +\nu_{it},\\ 
&\{\nu_{it}\}_{t} \indep  \{\alpha(U_i),\uyy_{i}(w)\} \\
&W_{it} = \mathbf{1}_{W^{\star}_{it} >0}
\end{aligned}
\end{equation}
Here $\alpha(U_i) = (\alpha_1(U_i),\alpha_2(U_i))\in \mathbb{R}^{p+1}$ captures the exposure of different units to aggregate shocks $\psi_t$ and $\nu_{it}$ represents an independent mean-zero idiosyncratic component. Because $W_{it}$ is binary, the model is nonlinear: instead of the latent index $W_{it}^{\star}$ we observe only $W_{it} =\mathbf{1}_{W^{\star}_{it} >0}$.

In this model endogeneity arises because the exposure $\alpha(U_i)$ can be correlated with the potential outcomes. This should be compared with a situation where $\alpha(U_i) = \alpha$, but $\nu_{it}$ can be correlated with $\uyy_{i}(w)$. In the latter case, $\psi_t$ can be used as an aggregate instrument. In our model the situation is different and we can use $\psi_t$ to control for $\alpha(U_i)$. To justify this formally we need to make two additional assumptions. First, we assume that $\nu_{it}$ is independent over $t$, second, we assume that $\nu_{it}$ has a logistic distribution. Together these two assumptions lead to the following model:
\begin{equation}\label{eq: logit_model}
\begin{aligned}
&\mathbb{E}[W_{it}| U_i] = \frac{\exp(\alpha_1(U_i) + \theta_t+\alpha_2^T(U_i)\psi_t)}{1+ \exp(\alpha_1(U_i) + \theta_t+ \alpha_2^T(U_i)\psi_t)},\\
&W_{it} \indep \{W_{il}\}_{l\ne t}\ \Bigl|\ U_i,
\end{aligned}
\end{equation}
which is a natural generalization of the familiar two-way logit model
\begin{equation}
\mathbb{E}[W_{it}| U_i] = \frac{\exp(\alpha_1(U_i)+ \theta_t)}{1+ \exp(\alpha_1(U_i)+\theta_t)}.\\
\end{equation}

Define the following statistic:
\begin{equation*}
S_i= \left(\sum_{t=1}^TW_{it}/T, \sum_{t=1}^T \psi_t W_{it}/T\right).
\end{equation*}
Intuitively, $S_i$  measures the exposure of unit $i$ to aggregate shocks. $S_i$ depends on observable quantities only and can be computed, and it is easy to show that the following independence condition holds:
\begin{equation}
U_i \ \indep\ \uww_i\ \Bigl|\ S_i
\end{equation}
and thus $S_i$ satisfies Assumption \ref{as:suf}.

Undoubtedly the assumptions that make $S_i$ a sufficient statistic are strong. In particular, the logistic distribution of $\nu_{it}$ is a functional form assumption on the unobserved idiosyncratic errors. Evaluation based on $S_i$ might be useful even if these assumptions are violated. First, one can interpret this strategy in the spirit of the recent work on grouped fixed effects (e.g., \cite{bonhomme2015grouped,bonhomme2017discretizing}) which is valid in large-$T$ regime under weak assumptions. Second, even if conditioning on $S_i$ does not remove all dependence between $\uww_i$ and $U_i$, we still expect it to capture essential aspects of the unobserved heterogeneity, making the doubly robust strategy that we propose next natural.

\subsubsection{Stationary dynamics}\label{sec:dyn_model}

In the previous example $W_{it}$ was mainly determined by aggregate exogenous shocks. In this section, we consider an opposite situation and assume that $W_{it}$ is mainly determined by its past. In particularly, suppose that $W_{it}$ satisfies the following first-order Markov model:
\begin{equation}\label{eq:alt_dyn_logit}
\begin{aligned}
&W_{it} \indep \{W_{il}\}_{l> t} \ \Bigl|\ U_i,W_{i}^{t-1},\\
&\mathbb{E}[W_{it}|U_i, W_i^{t-1}] = \mathbb{E}[W_{it}|U_i, W_{it-1}]\\
\end{aligned}
\end{equation}
Restricted versions of this model are commonly used to model dynamic binary choice decisions (see \cite{aguirregabiria2018sufficient} for examples). Let  $\pi_{i}(W_{it-1})\equiv\mathbb{E}[W_{it}|U_i, W_{it-1}]$ and observe that $\log\left(\frac{\pi_{i}(W_{it-1})}{1-\pi_{i}(W_{it-1})}\right)$ is a linear function of $W_{it-1}$:
\begin{equation}
\log\left(\frac{\pi_{i}(W_{it-1})}{1-\pi_{i}(W_{it-1})}\right) = \alpha(U_i) + \eta(U_i) W_{it-1}, 
\end{equation}
which can then be expressed in a familiar logit form:
\begin{equation}\label{eq:dynamic_model}
    \mathbb{E}[W_{it}|U_i, W_{it-1}] = \frac{\exp(\alpha(U_i) + \eta(U_i) W_{it-1})}{1+ \exp(\alpha(U_i) + \eta(U_i) W_{it-1})}.
\end{equation}

Define the following statistic:
\[S_i = \left(\sum_{t=2}^{T-1} W_{it}, \sum_{t=2}^{T} W_{it}W_{it-1}, W_{i1}, W_{iT}\right).\] 
It is not hard to show that, as long as conditions (\ref{eq:alt_dyn_logit}) are satisfied, $S_i$ is a sufficient statistic that satisfies Assumption \ref{as:suf}. In fact, the discussion above shows that  (\ref{eq:dynamic_model}) is equivalent to (\ref{eq:alt_dyn_logit}). This means that contrary to our previous example, in this model, sufficiency does not follow from a functional form assumption.

\subsubsection{Discussion}

Examples from Sections \ref{sec:ag_shocks} and \ref{sec:dyn_model} illustrate two different empirical settings in which one can use our approach. The first example emphasizes the role of exogenous aggregate shocks that are frequently used in applied literature to identify policy effects (e.g., \cite{duflo2007dams,dube2013commodity,nunn2014us, nakamura2014fiscal}). Our approach is applicable as long as the primary reason for endogeneity is differential exposure of different units to these shocks. The second example emphasizes the role of the structural assumptions, such as Markov restrictions, thus providing a way of combining structural choice models with the estimation of treatment effects.

From the formal point of view all the models considered above share the same structure. In all of them the conditional distribution of $\uww_i$ has the following representation:
\begin{equation}
\log\left(\mathbb{P}(\uww_i|U_i)\right) = S(\uww_i)^\top\alpha(U_i) + \beta(U_i) + \eta(\uww_i)
\end{equation}
In other words, the distribution of $W_i$ belongs to an exponential family, with $S(\uww_i)$ begin the sufficient statistic.  Sufficiency arguments have a long tradition in binary panel models literature (e.g., \cite{andersen1970asymptotic, chamberlain1984panel,honore2000panel,chamberlain2010binary}) where they have been used to obtain consistent estimators for common parameters. More recently, sufficiency arguments were used by \cite{aguirregabiria2018sufficient} to identify common parameters in dynamic structural models. We are using sufficiency differently: instead of using $S_i$ to identify common parameters, we use it to condition on unobserved heterogeneity, extending the results from \cite{arkhangelsky2018role} to panel models.

\section{Double robustness}\label{sec:dr}

In this section we build on our previous results and present a doubly robust identification argument. We then propose a natural algorithm that implements our strategy.

\subsection{Identification}

The sets of $\mmw_\outcome$ and $\mmw_\assignment$ described in Section \ref{sec:full_id} are motivated by different models and in general do not need to be similar. The first set is built with  within-unit across-period comparisons in mind, while the second one is based on within-period across-unit comparisons. The non-negativity constraints are also different: in $\mmw_\outcome$ we require every treated unit to have a non-negative weight, while in $\mmw_\assignment$ this property only holds for the weights averages within a subpopulation described by $S_i$. However, the sets may have common elements. 
In that case  one does not need to take a stand on what comparisons to use. Instead of using those validated by the outcome model or by the assignment model, one can use those that are validated by both.

Let $\mmw_\double = (\mmw_\outcome\cap\mmw_\assignment)$ and observe that combining the restrictions in (\ref{eq:out_1}) and (\ref{design}) we get that any $\weight \in \mmw_\double$ satisfies the following restrictions:
\begin{equation}
\mathrm{Target:}\, \frac{1}{T} \sum_{tk} \pi_k\weight_{kt} \mathbf{W}_{kt} = 1, \end{equation}
\begin{equation}
\mathrm{Within-unit\, balance:}\, \frac{1}{T} \sum_{t=1}^{T} \weight_{kt} = 0,\end{equation}
\begin{equation}
\mathrm{Within-period\, balance:}\,\sum_{k: \mathbf{W}_k \in \mathbf{W}^s}\pi_k\weight_{kt} = 0, \end{equation}
\begin{equation}
\mathrm{Non-negativity:}\,\weight_{kt} \mathbf{W}_{kt} \ge 0.
\end{equation}

The set $\mmw_\double$ treats units and periods asymmetrically. The weights need to balance arbitrary functions of $S_i$ within each period, but only need to balance unit fixed effects for every unit. Of course, this is a direct consequence of the two-way model that we consider for the outcomes. If the underlying model is more complicated -- e.g., there are interactive fixed effects -- then it will introduce additional restrictions. While we do not pursue such extensions in this paper, they can be addressed within our framework using the ideas from \cite{arellano2011identifying} and \cite{freyberger2018non}.

Combining earlier discussion of $\mmw_\outcome$ and $\mmw_\assignment$ it is easy to see that a necessary and sufficient condition for $\mmw_\double$ to be non-empty is that there exists a value $s$ for the sufficient statistic $S_i$ such that the corresponding $\mbw^s$ contains at least one of the following two sub-matrices (up to permutations):
\begin{equation}
\begin{aligned}
\mbw_1= \begin{pmatrix} 0 & 1\\ 0 & 0\end{pmatrix}
\ \ 
\mbw_2= \begin{pmatrix} 0 & 1\\ 1 & 0\end{pmatrix}
\\
\end{aligned}
\end{equation}
The requirement that for some $s$  the set  $\mbw^s$  contains at least one of these sub-matrices is in general more demanding than the assumption that $\mbw$ contains one of these submatrices. It is also more demanding than the overlap restriction in Assumption \ref{as:suf}.  At the same time, if $S_i$ includes $\overline W_i$ then for any $s$, $\mbw^s$ can contain $\mbw_1$ only if it contains $\mbw_2$ and this is equivalent to the overlap condition.


Finally we can state the main identification result. The following theorem is a direct consequence of Propositions \ref{pr:id_model} and \ref{pr:id_design}:
\begin{theorem}\label{th:identification}
Suppose Assumptions \ref{as:static}, \ref{as:unc} hold, and either
\ref{as:chamberlain}, or
\ref{as:suf}, or both hold. Then for any $\weight\in\mmw_\double$, the estimand $\tau(\weight)$ is a convex combination of treatment effects. 
\end{theorem}

It is important to compare this theorem with other doubly robust results in the panel literature (e.g., \cite{sant2020doubly}). The conventional interpretation of double robustness (e.g., \cite{robins1,chernozhukov2016double,chernozhukov2018double}) is based on two different ways of using covariates in estimation. One can either demean the outcomes, thus making units directly comparable, or re-weight the units to guarantee that the differences between them average out. To ensure good statistical properties (e.g., semiparametric efficiency), we can combine both of these ideas. In this case, the bias of the resulting estimator depends on the product of two errors. As a result, one can trade a less accurate outcome model for more precise weighting and vice versa. 

Our interpretation of double robustness is different, and in fact covariates do not play a role. We have not explicitly introduced the covariates, and thus any discussion on how to use them is complementary to our identification result. More precisely, our approach is based on combining two different identification arguments, whereas the traditional double robustness uses a single identification assumption (a version of conditional independence). In principle, one can combine traditional double robustness related to covariate adjustment with ours. 

\subsection{Algorithm}

Our estimator uses the panel data $\{Y_{it},W_{it},X_i\}_{i,t}$, where we now explicitly introduce time invariant covariates $X_i$. We assume that a researcher has constructed a sufficient statistic $S_i$. To incorporate covariates we consider two $p$-dimensional functions of $(X_{i},S_i,t)$ and $(X_i, S_i)$: $\psi^{(1)}(X_i, S_i,t) = (\psi^{(1)}_{1}(X_{i},S_i, t),\dots, \psi^{(1)}_{p}(X_{i},S_i, t))$,  $\psi^{(2)}(X_i,t) =(\psi^{(2)}_{1}(X_{i}, t),\dots, \psi^{(2)}_{p}(X_{i}, t))$, (where the second function does not depend on the sufficient statistic $S_i$) and define $\psi_{t}(X_i,S_i) \equiv (\psi^{(1)}(X_i,S_i,t),\psi^{(2)}(X_i,t))$. If $X_i$ is discrete these functions can be simply set to time-specific indicators for each value of $X_i$ and $S_i$.

Given these inputs, our estimator is defined in the following way:
\begin{equation}\label{eq:est_int}
\hat \tau \equiv\frac{1}{NT} \sum_{it}\hat \weight_{it} Y_{it}
\end{equation}
where the weights $\{\hat\weight_{it}\}_{it}$ solve the optimization problem:
\begin{equation}\label{primal_problem}
\begin{aligned}
\{\hat \weight_{it}\}_{it} = \argmin_{\{\weight_{it}\}_{it}} &\frac{1}{(NT)^2}\sum_{it} \weight_{it}^2 \\
\text{subject to: } &\frac{1}{nT}\sum_{it} \weight_{it} W_{it} = 1, & 
&\frac{1}{T}\sum_{i} \weight_{it}= 0,\\
&\frac{1}{N}\sum_{t} \weight_{it}= 0, &
&\frac{1}{NT}\sum_{it} \weight_{it} \psi_{t}(X_i,S_i) = 0,\\
& \weight_{it} W_{it} \ge 0,
\end{aligned}
\end{equation}
The weights $\hat \weight_{it}$ are related to  weights produced by OLS. The key difference is that we are explicitly looking for weights that balance out functions of $S_i$, not only fixed attributes $X_i$, and satisfy certain inequality constraints.

Our estimator fits naturally into recent theoretical literature on balancing weights (e.g., \cite{imai2014covariate,zubizarreta2015stable,athey2018approximate,hirshberg2017augmented,chernozhukov2018double,armstrong2018finite}). The main technical difference is that we need to balance unit-specific functions and explicitly impose non-negativity constraints. At the same time, we only balance a parametric class of functions of $(X_i, S_i)$, rather than a general nonparametric class (as in \cite{hirshberg2017augmented,armstrong2018finite})

The weights that we get from (\ref{primal_problem}) have the least possible norm, subject to balancing conditions motivated by Theorem \ref{th:identification}. Our choice of the objective function is justified by statistical reasons -- the variance of any linear estimator is directly related to the norm of the weights (even in heteroskedastic case). In practice, one can have additional concerns, such as interpretability of the resulting estimator. In this case, one can impose additional restrictions, for example by requiring the weights of the treated units to be uniform.

\section{Inference}\label{sec:est_inf}
\setcounter{equation}{0}

\subsection{Statistical framework}
We observe a random sample $\{(\uyy_{i}, \uww_i , X_i)\}_{i=1}^N$ where each $(\uyy_{i}, \uww_i , X_i)$ is distributed according to distribution $\mathcal{P}$.  For each unit we construct a sufficient statistics $S_i\equiv S( \uww_i , X_i)$, which includes $\overline W_{i}$.  In the analysis we focus on approximations with large $N$ and fixed $T$.

We maintain Assumptions \ref{as:static}, \ref{as:unc} and additionally restrict the outcome model to introduce covariates in a tractable way
\begin{assumption}\label{as:outcome_model_cov}
The following restriction holds for $t\in \{1,\dots, T\}$:
\begin{equation}
\begin{aligned}
    &Y_{it}(0) = \alpha_t(U_i) + \psi^{(2)}(X_i,t)^\top \delta^{(2)} + \nu_{it},\\
    &\mathbb{E}[\nu_{it}|U_i,X_i] = 0
\end{aligned}
\end{equation}
\end{assumption}
This assumption specifies a regression model that is separable in $X_i$ and $U_i$, and in addition imposes parametric restrictions on the second part. The second restriction can be relaxed in a standard way by the dimension of $\psi^{(2)}$ to grow with the sample size.

With our next assumption, we assume that  $\alpha_t(U_i)$ either satisfies the two-way model, or has a particular projection on  $(S_i,X_i)$.
\begin{assumption}\label{as:final_out_model}
If Assumption \ref{as:suf} is satisfied, then the following holds
\begin{equation}\label{eq:rest_suf}
\begin{aligned}
 &\alpha_t(U_i) = \beta_t +  \psi^{(1)}(X_i,S_i,t)^\top \delta^{(1)} + \upsilon_{it},\\
& \mathbb{E}[\upsilon_{it}|S_i,X_i] = 0.
\end{aligned}
\end{equation}
Otherwise, the two-way model holds: 
\begin{equation}\label{eq:rest_twfe}
    \alpha_t(U_i) =\beta_t + \alpha(U_i).
\end{equation}
\end{assumption}
We define the overall error
\begin{equation}
    \varepsilon_{it}\equiv \nu_{it} + \upsilon_{it},
\end{equation}
with the convention that $ \upsilon_{it} \equiv 0$ in case (\ref{eq:rest_twfe}) holds. By construction, under Assumptions \ref{as:unc} and \ref{as:final_out_model} this error has zero conditional mean:
\begin{equation}\label{eq:eps_prop}
    \mathbb{E}[\varepsilon_{it}|\uww_i,X_i] = 0.
\end{equation}

It is now easy to see that the estimator (\ref{eq:est_int}) has the following representation:
\begin{equation}
    \hat \tau = \frac{1}{NT}\sum_{it}\hat \gamma_{it} \tau_{it}W_{it} +  \frac{1}{NT}\sum_{it}\hat \gamma_{it} \varepsilon_{it}
\end{equation}
In our analysis we focus on the following (random) parameter:
\begin{equation}
\tau_{cond}=  \frac{1}{NT} \sum_{it} \hat\weight_{it}W_{it} \mathbb{E}[\tau_{it}|\uww_i, X_i]
\end{equation}
This is a conditional weighted average treatment effect, where the weights are directly observed and equal to $\hat \weight_{it}W_{it}$. By construction these weights are nonnegative and thus $\tau_{cond}$ is a convex combination of individual and time-specific treatment effects. 

Using this definition we expand $\hat \tau$ further:
\begin{equation}\label{eq:simp_decomp}
    \hat \tau - \tau_{cond} =  \frac{1}{NT}\sum_{it}\hat \gamma_{it} u_{it}W_{it} +  \frac{1}{NT}\sum_{it}\hat \gamma_{it} \varepsilon_{it},
\end{equation}
where $ u_{it}\equiv \tau_{it} - \mathbb{E}[\tau_{it}|\uww_i,X_i]$. This expansion together with (\ref{eq:eps_prop}) shows that the estimator is unbiased for $\tau_{cond}$. This is a direct consequence of Assumptions \ref{as:outcome_model_cov} and \ref{as:final_out_model}. In a less restrictive model (\ref{eq:simp_decomp}) would include a bias term which could then be bounded by invoking standard smoothness assumptions. In our analysis below we focus on the distributional properties of (\ref{eq:simp_decomp}) in large samples.

\subsection{Formal results}
To guarantee that the weights $\hat\gamma_{it}$ are well-behaved in large samples we need to impose restrictions on $(W_i,S_i,X_i)$. Define $\psi_{it} = (\mathbf{1}_{t =1},\dots, \mathbf{1}_{t =T}, \psi_t(X_i,S_i))$ -- $(T+2p)$-dimensional random vector. We consider its demeaned version:
\begin{equation}
\Gamma_{it}\equiv(1-W_{it})\psi_{it} - \frac{\sum_{l=1}^T(1-W_{il})\psi_{il}}{\sum_{l=1}^T(1-W_{il})},
\end{equation}
with the convention that $\Gamma_{it} \equiv 0$ if $W_{it} = 1$ for all $t$. For each period this vector measures the unit-level difference between $\psi_{it}$ and the average of $\psi_{it}$ restricted to periods where unit $i$ is treated. We impose the following assumptions on the joint distribution of $(W_i,S_i,X_i)$:
\begin{assumption}\label{as:tech_1}
(a) For arbitrary $\eta\in \mathbb{R}^{T+2p}$ and $\kappa >0$ the following holds:
\begin{equation}
\begin{aligned}
& \eta^\top\left(\frac{1}{T}\sum_{t=1}^T\mathbb{E}\left[ (1-W_{it})\Gamma_{it}\Gamma_{it}^{\top}\right]\right)\eta \ge \kappa \|\eta\|_2^2;
\end{aligned}
\end{equation}
(b) $\mathcal{P}$-a.s. $(X_i, S_i)$ belongs to a compact set, and $\psi(X_i, S_i, t)$ is a continuous function on this set.
\end{assumption}
The second part of this assumption is more restrictive then necessary but we expect it to hold in applications and it simplifies the formal analysis. Our next assumption restricts $S_i$ and imposes an overlap condition.
\begin{assumption}\label{as:tech_2}
$S_i$ includes $\overline W_i$ and $\mathbb{E}[W_{it}|S_i, X_i] <1-\eta$ for some $\eta>0$ $\mathcal{P}$ a.s.
\end{assumption}
In practice we expect $\overline W_i$ always be part of the sufficient statistic, while the second condition is a mild overlap restriction that we use to guarantee that the problem has a solution. Finally, to guarantee good behaviour of our estimator in the large samples we impose standard moment restrictions on $\epsilon_{it}\equiv \varepsilon_{it} + W_{it} u_{it}$.
\begin{assumption}\label{as:tech_3}
Errors $\varepsilon_{it}$ satisfy moment conditions:
\begin{equation}
\begin{aligned}
&0 <\underline {\sigma}^2_\epsilon\le \mathbb{E}[\epsilon^2_{it}|\uww_i, X_i] \le \overline {\sigma}^2_\epsilon < \infty,\\
&\mathbb{E}[\epsilon^4_{it}] < \infty.
\end{aligned}
\end{equation}
\end{assumption}

Next theorem shows that $\hat \tau$ exists with large probability and is close to $\tau_{cond}$:
\begin{theorem}\label{th:main_inf_theorem}
Suppose Assumptions \ref{as:static},\ref{as:unc}, \ref{as:outcome_model_cov},\ref{as:final_out_model},\ref{as:tech_1},\ref{as:tech_2},\ref{as:tech_3} hold. Then (a) with probability approaching one $\hat \tau$ exist; (b) there exist random variables $\left\{\weight^{\star}_t(X_i, \uww_i)\right\}_{t=1}^T$ such that the following holds:
\begin{equation}
\frac{1}{T}\sum_{t=1}^T\|\hat \weight_t - \weight^{\star}_t\|_2 = o_p(1);
\end{equation}
(c) the estimator is asymptotically normal
\begin{equation}
\sqrt{n}\left(\hat \tau - \tau_{cond}\right) \rightarrow \mathcal{N}(0, \sigma^2),
\end{equation}
where $\sigma^2=\mathbb{E}\left[ \left(\frac{1}{T}\sum_{t=1}^T\weight^{\star}_{it}\epsilon_{it}\right)^2\right]$.
\end{theorem}
This theorem describes the performance of our estimator in large samples. The population weights $\weight^{\star}$ depend on $(X_i, \uww_i)$, not only on $S_i$ which is an implication of the fact that we need to deal with individual fixed effects. 

To conduct inference we need to construct an estimator for $\sigma^2$. Our next results shows that conventional unit-level bootstrap provides a conservative estimator for the variance.
\begin{theorem}\label{th:bootstrap}
Let $\{\hat\tau_{(b)}\}_{b=1}^B$ be a set of non-parametric (unit-level) bootstrap analogs of $\hat \tau$. Define:
\begin{equation}
\hat \sigma^2 \equiv \frac{N}{B} \sum_{b=1}^B \left(\hat \tau_{(b)} -\hat \tau\right)^2
\end{equation}
and suppose that assumption of Theorem  \ref{th:main_inf_theorem} hold. Then if $ \mathbb{E}[\tau_{it}|\uww_i, X_i] = \tau$ $\hat \sigma^2$ is consistent for $\sigma^2$; otherwise $\hat \sigma^2$ is conservative. 
\end{theorem}
Theorems \ref{th:main_inf_theorem} and \ref{th:bootstrap} imply that one can construct asymptotically conservative confidence intervals by standard methods. In particular, let $z_{\alpha}$ be an $\alpha$-level quantile of the standard normal distribution. Then the following interval has an asymptotic coverage of at least $1-\alpha$:
\begin{equation}
    \tau_{cond} \in \hat \tau \pm \sqrt{\frac{\hat \sigma^2}{N}}z_{\alpha/2}
\end{equation}

\section{Non-binary Treatments}\label{sec:extensions}
In applications, the treatment $W_{it}$ is often non-binary, and the results discussed so far are not directly applicable. One possibility is to binarize the treatment, but this process will change both the outcome and the assignment models. This section discusses an alternative strategy for dealing with a general treatment. 

To proceed we need to specify the outcome model and the assignment process for general non-binary treatment $W_{it}$. For the outcome model we resort to the two-way linear structure:
\begin{equation}
\begin{aligned}
&Y_{it}(w) = \alpha(U_i) + \lambda_t + \tau_{t}(U_i) w +\epsilon_{it}\\
&\mathbb{E}[\epsilon_{it}|U_i] = 0\\
\end{aligned}
\end{equation}
thus abstracting away from  potential non-linear effects of $w$. This is the standard assumption made in applications, and it does not take us far from the current empirical practice.

For the assignment model we consider a baseline distribution $f_0(w)$ that has the same support as $W_{it}$. If $W_{it}$ is non-negative, then this can be an exponential distribution, if $W_{it}$ represents counts of certain events, then $f_0(w)$ can be Poisson. We then assume that the distribution of $W_{i}$ conditional on $U_i$ belongs to the following exponential family:
\begin{equation}\label{eq:glm_basic}
    f(W_i|U_i) = \exp\left\{\sum_{t}\beta^\top(U_i) \psi_t(W_{it}) -\psi(U_i)\right\}\prod_{t}f_0(W_{it})
\end{equation}
where $\psi_t(\cdot)$ is a known function. This structure directly generalizes the example presented in Section \ref{sec:ag_shocks}. In particular, if we observe aggregate shocks $Z_{t}$ then it is natural to consider $\psi_t(W_{it}) = Z_{t}W_{it}$.

Exponential structure of the assignment model implies the general unconfoundedness condition:
\begin{equation}
    W_{i}\indep \{Y_i(w)\}_{w}| S_i
\end{equation}
where $S_i = \sum_t\psi_t(W_{it})$. Given $S_i$ we can proceed by identifying the effect by running the standard two-way regression:
\begin{equation}
    Y_{it} = \alpha_i + \lambda_t + \tau_{it}W_{it} + \epsilon_{it}
\end{equation}
withing the subpopulations defined by $S_i$. This approach delivers meaningful causal effects if $\tau_{it}$ does not vary in the subpopulations defined by $S_i$. In practice, we can split the data into clusters with similar values of $S_i$, run OLS separately for each cluster, and then aggregate the effects. This approach connects us with the recent work on fixed-effect models  (\cite{bonhomme2015grouped,bonhomme2017discretizing}), where the authors argue for using K-means algorithm to classify units into clusters, as a way to estimate computationally challenging fixed-effect models. Our results shows how to use the assignment model to derive the characteristics that can be used for such classification. 

The model (\ref{eq:glm_basic}) can be considerably generalized. Instead of the aggregate shocks, one can consider stationary dynamic models from Section \ref{sec:dyn_model}. In general, one can use a rich class of generalized linear models (e.g., \cite{nelder1972generalized,efron2016computer}) to adapt the assignment process to the particular structure of $W_{it}$. These models are commonly used in applied data analysis to understand complex data structures, and our results show how to exploit them for identification purposes. 

If we use the normal distribution as the baseline for $W_{it}$ then the assignment model reduces to linear regression:
\begin{equation*}
    W_{it} = \beta(U_i)^\top \psi_t + \nu_{it}
\end{equation*}
implying that $W_{it}$ can be decomposed into a low-rank component $\beta(U_i)^\top \psi_t$ and idiosyncratic noise $\nu_{it}$.  This suggests that one can use interactive fixed effects regressions to estimate the treatment effects. This choice is attractive for some applications, but the panels that we have in mind have small $T$ and large $n$, rendering both conventional interactive fixed effects regressions (e.g., \cite{bai2009panel}) and its regularized analogs (e.g., \cite{chernozhukov2019inference}) inconsistent. Also, we do not restrict the persistence in the errors of the potential outcomes; thus, the GMM estimators in the spirit of \cite{holtz1988estimating}, or \cite{freyberger2018non} are inapplicable as well.

\section{Empirical Illustration and Simulations}\label{sec:emp_il}
\setcounter{equation}{0}

\subsection{Empirical illustration}
To illustrate our approach at work in a real application we consider data from \cite{charles2013employment}. In the paper authors analyze the relationship between the local voting preferences (expressed by turnout) and local economic outcomes (such as earnings or employment). In particular, the stylized version of the main regressions that is proposed in the paper has the following form:
\begin{equation}\label{eq:basic_reg}
    Y_{it} = \alpha_i + \lambda_t + \tau W_{it} + \epsilon_{it}
\end{equation}
where the unit of observation $i$ corresponds to the U.S. counties, $Y_{it}$ measures the local turnout (we will focus on the presidential elections), $W_{it}$ measure the log-income per capita in the corresponding county.  Authors estimate $\tau$ by IV, using aggregate shocks to construct instruments. In particular, their first stage model has the following form:
\begin{equation}\label{eq:basic_fs}
\Delta W_{it} =  \theta_t + \gamma^\top_1D_{1i}\Delta Z_{1t} +  \gamma^\top_2D_{2i} \Delta Z_{2t} + \nu_{it}
\end{equation}
where $\{Z_{1t},Z_{2t}\}_t$ correspond to nation-level oil and coal prices, and $D_{ki} = (D_{1ki},D_{2ki})$ are indicators for the importance of oil and coal for the county $i$ (medium or large). As a result, authors use the variation in $\Delta W_{it}$ that is ``cleaned'' from $\nu_{it}$ and $\theta_t$ to identify $\tau$. This variation is coming from two sources: the variation in $D_{ki}$ over $i$ and $\Delta Z_t$ over $t$. The underlying identification assumption behind this approach is that the endogeneity problem arises from $\nu_{it}$ being correlated with $\Delta \epsilon_{it}$, while $D_{ki}$ is not.

Our approach to identification is different: instead of (\ref{eq:basic_fs}) we consider the following first stage model:
\begin{equation}\label{eq:our_fs}
    W_{it} = \beta_i + \theta_t + \gamma_{1i}Z_{1t} + \gamma_{2i}Z_{2t} + \upsilon_{it}
\end{equation}
and assume that $(\beta_i, \gamma_{1i},\gamma_{2i})$ are correlated with the potential outcomes, while $\{\upsilon_{it}\}_{t}$ are not. Using our previous notation, we can express this in the following way:
\begin{equation}
    U_i =(\beta_i, \gamma_{1i},\gamma_{2i})
\end{equation} 
As a result, our approach is complimentary to that of \cite{charles2013employment}. While they exploit the variation in $D_{ki}$ -- which can be viewed as a proxy for $\gamma_{ki}$ -- we instead control for it and exploit the variation in $\upsilon_{it}$.\footnote{In principle one can utilize variation in $Z_{kt}$ only, but more time periods are needed for this approach to be practically useful. See \cite{arkhangelsky2020policy} for more details.}

We use (\ref{eq:our_fs}) to construct the sufficient statistic for $U_i$:
\begin{equation}
S_i\equiv \left(\sum_{t\le T} W_{it},\sum_{t\le T} Z_{1t}W_{it}, \sum_{t\le T} Z_{1t}W_{it}    \right)
\end{equation}
Note that if $\{\nu_{it}\}_t$ has normal distribution then $S_i$ is sufficient for $U_i$, otherwise one can justify using $S_i$ with the logic from  Section \ref{sec:ag_shocks}. $S_i$ is a 3-dimensional object and to control for it we split them into $K=1000$ groups with similar values of $S_i$. Once these groups are defined, we proceed by estimating \ref{eq:basic_reg} by OLS with two-way fixed effects within each group. The results of the estimation are presented in Table \ref{table:emp_res}. These results are qualitatively similar to those obtained by  \cite{charles2013employment} who also do not find significant effects for the presidential elections. 

\begin{table}[t]
\begin{center}
\begin{tabular}{l|r|r}
  \hline
 & estimate & s.e. \\ 
  \hline
$\hat \tau_{DR}$ & 0.003 & 0.007 \\ 
   \hline
\end{tabular}
\caption{The results are based on the data from $n = 2994$ counties over $T = 8$ presidential elections (1972-2000). The outcome is the turnout in the presidential elections at the county level, and the treatment is the $\log$-earnings. Sufficient statistic $S_i$ is constructed using $\log$(national employment) in coal and gas industries. $K$-means algorithm is used to split counties into $K = 1000$ groups based on $S_i$. Standard errors are computed using county-level bootstrap.
\label{table:emp_res}}
\end{center}
\end{table}

\subsection{Simulations}
We use the data from  \cite{charles2013employment} as a basis for a simulation. Let $\mathbf{Y}$ and $\mathbf{W}$ be $n\times T$ matrices with enties equal to $Y_{it}$ and $W_{it}$, respectively. We standardize each of these matrices by subtracting the overall mean, and dividing by the overall standard deviation. We then decompose them into three components:
\begin{equation}
    \begin{aligned}
    &\mathbf{Y} = \mathbf{F}_{Y} + \mathbf{L}_{Y} + \mathbf{E}_{Y}\\
     &\mathbf{W} = \mathbf{F}_{W} + \mathbf{L}_{W} + \mathbf{E}_{w}
    \end{aligned}
\end{equation}
where $F_{k} = \alpha_i^{(k)} + \lambda_t^{(k)}$ is the two-way matrix, $L_{k}$ is a matrix of rank $5$, and $\mathbf{E}_{k}$ captures the residual variation. We compute the size and the correlation between the residuals and elements of matrices $L_k$:
\begin{equation}
    \begin{aligned}
    &\sigma^2_{k}(E) = \frac{\sum_{it}E^2_{k,it}}{nT}\\
    &\rho(E) = \frac{\sum_{it}E_{y,it}E_{w,it}}{nT\sigma_w(E) \sigma_y(E)}\\
    &\sigma^2_{k}(L) = \frac{\sum_{it}L^2_{k,it}}{nT}\\
    &\rho(L) = \frac{\sum_{it}L_{y,it}L_{w,it}}{nT\sigma_w(L) \sigma_y(L)}
    \end{aligned}
\end{equation}
and then simulate the data using the following model:
\begin{equation*}
    \begin{aligned}
    &Y_{it} = F^{(b)}_{Y,it} + \frac{1}{c(\zeta)}\left(\sqrt{(1-\zeta^2)} L^{(b)}_{Y,it} + \zeta\frac{\sigma_{Y}(L)}{\sigma_{W}(L)}L^{(b)}_{W,it}\right) + \epsilon_{it}\\
    &W_{it} = F^{(b)}_{W,it} + L^{(b)}_{W,it} + \upsilon_{it}\\
    \end{aligned}
\end{equation*}
where $(F_{Y,it}^{(b)}, L_{Y,it}^{(b)},F_{W,it}^{(b)}, L_{W,it}^{(b)})$ are sampled uniformly from the rows of $F_k,L_k$, while $(\epsilon_{it},\nu_{it})$ have a joint normal distribution with the covariance matrix implied by $\sigma_y,\sigma_w,\rho$. Parameter $\zeta$ controls the excess selection bias that is not present in the real data. Parameter $c(\zeta)$ normalizes this component to have the expected sum of squares equal to $\sigma^2_Y(L)$ to keep the relative sizes of the fixed effects and the low-rank component constant. We consider two designs: in the first one $\zeta = 0$, in the second it is equal to $0.05$. Note that in this simulation the effect of the treatment is equal to zero, which is natural given the results presented in the previous section. 

The summary of the results over 1000 simulations is presented in Table \ref{table:sim_res}. We use two benchmarks: the standard two-way OLS regression and the TSLS regression implemented in the original paper. In the baseline case, our estimator and the standard TW perform equally well. Both exhibit certain bias, which is not surprising given the presence of the low-rank components and the correlation between the errors. Once we introduce the additional selection bias $(\zeta = 0.05)$, the performance of TW estimator deteriorates considerably (1300\% increase in RMSE), while our estimator continues to perform well (50\% increase in RMSE).

We emphasize that we do not generate the treatment using the first stage described in the previous section:
\begin{equation*}
    W_{it} = \theta_t + \gamma_{1i}Z_{1t} + \gamma_{2i}Z_{2t} + \upsilon_{it}
\end{equation*}
Instead, we use the actual data to extract the systematic components of $W_{it}$. Moreover, we do not set the correlation between $\upsilon_{it}$ and $\epsilon_{it}$ to zero. There are two reasons for the success of our estimator. First, the contribution of the errors $\upsilon_{it},\epsilon_{it}$ to the corresponding outcomes is small compared to the contribution from $L_{k,it}$ and $F_{k,it}$. In particular, in both cases, the two-way fixed effects play a key role, explaining $85\%$ of the variation in the data. This drives the good behavior of both TW and DR when $\zeta = 0$. Once we scale the correlation between $L_{W, it}$ and $L_{Y,it}$ the low-rank component starts to play a role, and TW estimator does not do anything about it. In contrast, the aggregate shocks $(Z_{1t},Z_{2t})$ allow us to extract important components of the low-rank matrix and control for them.

\begin{table}[t]
\begin{center}
\begin{tabular}{|l|rr|rrr|rrr|}
\hline
&  \multirow{2}{*}{$\rho(L)$} &   \multirow{2}{*}{$\rho(E)$} & \multicolumn{3}{c|}{RMSE}& \multicolumn{3}{c|}{Bias}\\
& && DR & TW & TSLS & DR & TW &TSLS\\
\hline 
 \hline
 Design 1 ($\zeta = 1$) & 0.039& -0.038 &0.016 &   0.016 &    0.521 &  -0.01  &  0.013 &   -0.005\\
  Design 2 ($\zeta = 0.95$) & 0.351 & -0.038 & 0.023   & 0.205 &    0.505 &   0.02  &  0.205  &   0.222\\
  \hline
\end{tabular}
\caption{Results are based on $1000$ simulations, with  $n = 2994$ and $T = 8$; size of the outcome and assignment models components: $(\sigma^2_Y(L),\sigma^2_W(L),\sigma^2_Y(E),\sigma^2_W(E)) =  (0.12,  0.09,0.02,0.06)$.
  \label{table:sim_res}}
 \end{center}
  \end{table}						       

\section{Conclusion}\label{sec:conc}
\setcounter{equation}{0}

This paper proposes a novel identification argument that can be used to evaluate a causal effect using panel data. We show that one can naturally combine familiar restrictions on the relationship between the outcome and the unobserved unit-level characteristics with reasonable economic models of the assignment. Our approach allows us to construct a doubly robust identification argument: our estimand has causal interpretation if either the outcome model is correct or the assignment model is correct (or both). Using these results, we construct a natural generalization of the standard two-way fixed effects estimator robust to arbitrary heterogeneity in treatment effects. We prove that our estimator is asymptotically normal in large samples and show how to use it to conduct asymptotically valid inference.

\bibliographystyle{chicago}
\bibliography{references}

\begin{appendix}
\section{Appendix}
\setcounter{equation}{0}

\subsection{Propositions}
\textbf{Proof of Proposition \ref{pr:id_model}}:
For any $\weight \in \mmw_\outcome$ we defined the random variables
\begin{equation}
\begin{aligned}
&\weight_{it} \equiv \sum_{k=1}^{K} \weight_{kt} \{\uww_{i} = \mbw_{k}\}\\
\end{aligned}
\end{equation}
and considered the following estimator:
\begin{equation}
\tau(\weight) = \mme\left[\frac{1}{T}\sum_{t=1}^TY_{it} \weight_{it} \right]
\end{equation}
By assumption we have the representation:
\begin{multline}
 \mme\left[\frac{1}{T}\sum_{t=1}^TY_{it} \weight_{it} \right] =  \mme\left[\frac{1}{T}\sum_{t=1}^T(\alpha(U_i) + \lambda_t + \tau_t(U_i)W_{it} + \varepsilon_{it}) \weight_{it} \right] = \\
  \mme\left[\frac{1}{T}\sum_{t=1}^T(\alpha(U_i) + \lambda_t + \tau(U_i)W_{it} + \varepsilon_{it})  \sum_{k=1}^{K} \weight_{kt} \{\uww_{i} = \mbw_{k}\} \right] \\
  =  \mme\left[\frac{1}{T}\sum_{t=1}^T(\alpha(U_i)\weight_{kt} \{\uww_{i} = \mbw_{k}\}) \right] +\\
   \frac{1}{T}\sum_{t=1}^T \lambda_t\sum_{k=1}^K\mme\left[\weight_{kt} \{\uww_{i} = \mbw_{k}\} \right] +   \mme\left[\frac{1}{T} \sum_{k=1}^{K}\sum_{t=1}^T \tau_{t}(U_i)\{\uww_{i} = \mbw_{k}\}\mbw_{kt} \weight_{kt}  \right] = \\
 \frac{1}{T}\sum_{t=1}^T\lambda_t \sum_{k=1}^{K} \pi_k \weight_{kt} +  \sum_{t=1}^T\mme\left[\tau_t(U_i)\xi_t(\uww_i)  \right] =  \sum_{t=1}^T \mme\left[\tau_t(U_i)\xi_t(\uww_i)  \right]
\end{multline}
where $\xi_t(\uww_i)\equiv \sum_{k=1}^{K}\{\uww_{i} = \mbw_{k}\} \mbw_{kt} \weight_{kt}\ge0$. The first equality follows from the restrictions on the outcome model, the second -- by definition of the weights, the third -- because $\mathbb{E}[\varepsilon_i|U_i] = 0$ and strict exogeneity assumption; finally the last two equalities follow by construction of weights. By construction we also have that $\xi_t(\uww_i)\ge0$ and $\mathbb{E}[\sum_{t=1}^T\xi_t(\uww_i)] = 1$.\\

\textbf{Proof of Proposition \ref{pr:op_unc}}: For arbitrary $\underline w$ and measurable $A_0, A_1$ we need to demonstrate conditional independence::
\begin{multline}
\mathbb{E}[\{\uww_i = \underline w\}\{\underline Y_{i}(0) \in A_0, \underline Y_{i}(1) \in A_1\}|S_i] =\\ \mathbb{E}[\{\uww_i = \underline w\}|S_i]\mathbb{E}[\{\underline Y_{i}(0) \in A_0, \underline Y_{i}(1) \in A_1\}|S_i]
\end{multline}
Consider a chain of equalities:
\begin{multline}
\mathbb{E}[\{\uww_i = \underline w\}\{\underline Y_{i}(0) \in A_0, \underline Y_{i}(1) \in A_1\}|S_i] = \\
\mathbb{E}[\mathbb{E}[\{\uww_i = \underline w\}\{\underline Y_{i}(0) \in A_0, \underline Y_{i}(1) \in A_1\}|S_i,U_i]|S_i] = \\
\mathbb{E}[\mathbb{E}[\{\uww_i = \underline w\}|S_i,U_i]\mathbb{E}[\{\underline Y_{i}(0) \in A_0, \underline Y_{i}(1) \in A_1\}|S_i,U_i]|S_i] =\\
\mathbb{E}[\mathbb{E}[\{\uww_i = \underline w\}|S_i]\mathbb{E}[\{\underline Y_{i}(0) \in A_0, \underline Y_{i}(1) \in A_1\}|S_i,U_i]|S_i] =\\
 \mathbb{E}\{\uww_i = \underline w\}|S_i]\mathbb{E}[\{\underline Y_{i}(0) \in A_0, \underline Y_{i}(1) \in A_1\}|S_i]
\end{multline}
where the second equality follows by strict exogeneity and the fact that $S_i$ is a measurable functionon $\uww_i$, and the third one follows from sufficiency.\\

\textbf{Proof of Proposition \ref{pr:id_design}}:
The proof is similar to Proposition \ref{pr:id_model} and is omitted. 

\subsection{Dual representation}\label{sec:dual}
The Lagrangian saddle-point problem for the program (\ref{primal_problem}) has the following form:
\begin{multline}
\inf_{\weight_{it}}\sup_{\lambda_{(t)}, \lambda_{(i)}, \eta, \mu_{it}\ge 0, \pi\ge 0}\frac{1}{(NT)^2}\sum_{it} \weight_{it}^2  + \frac{1}{N} \sum_{i} \lambda_{(i)}\left(\frac{1}{T}\sum_{i} \weight_{it}\right) +\\
\frac{1}{T} \sum_{t}\lambda_{(t)} \left(\frac{1}{N}\sum_{t} \weight_{it}\right) +\pi\left(1-\frac{1}{NT}\sum_{it} \weight_{it} W_{it} \right) -\\
 \eta^\top\left(\frac{1}{NT}\sum_{it} \weight_{it} \tilde\psi_{it}\right) - \frac{1}{NT}\sum_{it}\mu_{it}\weight_{it}W_{it}
\end{multline}
where we use $\tilde\psi_{it}$ as a shorthand for $\psi(X_i, S_i, t)$. In Lemma \ref{lemma:ex_sol} we show that strong duality holds and we can rearrange the minimization and maximization: 
\begin{multline}
\sup_{\lambda_{(t)}, \lambda_{(i)}, \eta,\mu_{it}\ge 0, \pi\ge 0} \inf_{\weight_{it}}\frac{1}{(NT)^2}\sum_{it} \weight_{it}^2 + \frac{1}{N} \sum_{i} \lambda_{(i)}\left(\frac{1}{T}\sum_{i} \weight_{it}\right) + \\
\frac{1}{T} \sum_{t} \lambda_{(t)}\left(\frac{1}{N}\sum_{t} \weight_{it}\right) - \pi\left(\frac{1}{NT}\sum_{it} \weight_{it} W_{it} -1\right) - \\
\eta^\top\left(\frac{1}{NT}\sum_{it} \weight_{it} \tilde\psi_{it}\right) - \frac{1}{NT}\sum_{it}(\mu_{it}\weight_{it}W_{it})
\end{multline}
Solving this in terms of $\weight_{it}$ (an unconstrained quadratic problem) we get the following representation:
\begin{equation}
\inf_{\lambda_{(t)}, \lambda_{(i)}, \eta,\mu_{it}\ge 0, \pi\ge 0}\mathbb{P}_n\left[\frac{1}{T}\sum_{t=1}^T\left( \pi W_{it} -\lambda_{(t)}-\lambda_{(i)} -\eta^\top \tilde\psi_{it} - \mu_{it}W_{it}\right)^2\right] 
- \frac{4\pi}{N}
\end{equation}
We can further simplify this expression by concentrating out $\mu_{it}$ and $\pi$. To this end, define the following loss function:
\begin{equation}
\rho_z(x)\equiv x^2(1-z) + x_{+}^2z
\end{equation}
After some algebra we get the following:
\begin{equation}
\inf_{\lambda_{(t)}, \lambda_{(i)}, \eta}  \mathbb{P}_n\left(\frac{1}{T}\sum_{t=1}^T\rho_{W_{it}}\left(W_{it} - \lambda_{(t)}-\lambda_{(i)}- \eta^\top \tilde\psi_{it}\right)\right)
\end{equation}
Let $\{\hat\lambda_{(t)}, \hat \lambda_{(i)},\hat\eta\}_{i,t}$ be the solutions to this problem. The optimal unnormalized weights are equal to the following:
\begin{equation}
\hat \weight_{it}^{(un)} =\left(W_{it} - \hat\lambda_{(t)}-\hat\lambda_{(i)}- \hat\eta^\top \tilde\psi_{it}\right)(1-W_{it}) + 
\left(W_{it} - \hat\lambda_{(t)}-\hat\lambda_{(i)}- \hat\eta^\top \tilde\psi_{it}\right)_{+}W_{it}
\end{equation}
and the optimal weights are given by the normalization:
\begin{equation}
\hat \weight_{it}\equiv \frac{\hat \weight_{it}^{(un)}}{\frac{1}{NT}\sum_{it} \hat \weight_{it}^{(un)} W_{it}}
\end{equation}
By construction the weights are non-negative for the treated units and sum up to one once multiplied by $W_{it}$. The denominator is strictly positive under the conditions of Lemma \ref{lemma:ex_sol}.

\subsection{Lemmas}

In the following analysis we use $\psi_{it}$ (or $\psi_t$ for a generic value) as a shorthand for $T+2p$-dimensional vector that includes time fixed effects and all relevant covariates. For arbitrary $\eta$ and $\uww$ such that $\overline{W} <1$ define a function $g(X, \uww, \eta)$:
\begin{equation}\label{eq:g_def}
g(X, \uww, \eta)= \argmin_{\alpha}\left\{\frac{1}{T} \sum_{t=1}^T \rho_{W_t}(W_{t} - \alpha - \psi^{\top}_{t}\eta)\right\}
\end{equation}
The following lemma shows that $g$ is well-defined, and demonstrates that it is well-behaved in $\eta$ as long as $\psi_t$ is bounded.
\begin{lemma}\label{lemma:prop_ex}
Function $g(X, \uww, \eta)$ is uniquely defined, if $\|\psi_t\|_{\infty}<K$ then $g(X, \uww, \eta)$ is $\mathcal{P}$ a.s. uniformly (in $(X, \uww)$) Lipschitz in $\eta$. 
\end{lemma}
\begin{proof}
If $\overline W <1$ then the minimized function is strictly convex with a unique minimum. Define $h_t\equiv W_t-\psi_t^{\top}\eta$; and let $\tilde h_{(1)}, \dots, \tilde h_{(\sum_{t=1}^TW_t)}$ be the decreasing ordering of $h_t$ for units with $W_{t} = 1$; let $\tilde h_{(0)} = 0$. For $k = 0, \dots, \sum_{t=1}^TW_t$ define the following functions:
\begin{equation}
g_{k}(X, \uww, \eta) \equiv  \frac{\sum_{t=1}^T(1-W_{it})h_t + \sum_{l=0}^k \tilde h_{(l)}}{\sum_{t=1}^T(1-W_{it}) +k}\\
\end{equation}
It is easy to see that we have the following:
\begin{equation}
g(X, \uww, \eta) = g_{0}(X, \uww, \eta) + \sum_{l=1}^{\sum_{t=1}^TW_t}  \{\tilde h_{(l)}\ge g_{(l-1)}\}(g_{l}(X, \uww, \eta) - (g_{l-1}(X, \uww, \eta) )
\end{equation}
From this representation if follows that $g(X, \uww, \eta)$ is locally linear in $\eta$, and thus if $\psi_t$ is $\mathcal{P}$-a.s. bounded, then it is uniformly (in $(X, \uww)$) Lipschitz in $\eta$. 
\end{proof}

For future use we define $g(X, \uww, \eta) \equiv \infty$ for $\uww$ such that $\overline W = 1$. It is straightforward to see that with this definition (\ref{eq:g_def}) holds for all $\uww$ (of course, minimum is not unique if $\overline W = 1$).

\begin{lemma}\label{lemma:exist}
Suppose Assumptions \ref{as:tech_1} are satisfied. Then there exist a vector $\eta^{\star}\in \mathbb{R}^{T+2p}$ such that the following conditions are satisfied:
\begin{equation}
\begin{aligned}
&\xi_{it} \equiv W_{it} - g(X, \uww, \eta^{\star}) - \psi_{it}^\top \eta^{\star},\\
&\mathbb{E}\left[ \sum_{t=1}^T\xi_{it}\psi_{it}(1-W_{it}\{W_{it} \le g(X, \uww, \eta^{\star}) +\psi_{it}^{\top}\eta^{\star}\})\right] = 0,\\
&\sum_{t=1}^T \xi_{it}(1-W_{it}\{W_{it} \le g(X, \uww, \eta^{\star}) +\psi_{it}^{\top}\eta^{\star}\})= 0.
\end{aligned}
\end{equation}
\end{lemma}
\begin{proof}
Define $\mathcal{F}\equiv \{f = (f_{1},\dots f_{T}): f_{t} = h_0 +\psi_{t}^\top \eta, h_0 \in L_{2}(\mathcal{P}_{mar}), \eta \in \mathbb{R}^{2p}\}$. For any $f\in \mathcal{F}$ define the risk function:
\begin{equation}
R(f) \equiv \mathbb{E}\left[\frac{1}{T} \sum_{t=1}^T \rho_{W_{it}}(W_{it} -f_t(\uww_i, X_i))\right]
\end{equation}
By construction it satisfies the following bound:
\begin{equation}
    R(f) \ge \eta^\top\mathbb{E}\left[\frac{1}{T} \sum_{t=1}^T(1-W_{it})\Gamma_{it}\Gamma_{it}^\top \right]\eta \ge \kappa \|\eta\|_2^2
\end{equation}
which implies that we can focus our attention to $\|\eta\| \le C$, where $C$ is some constant. 

Using Lemma \ref{lemma:prop_ex} we restrict attention to $f \in \tilde{\mathcal{F}} = \{f \in \mathcal{F}:f_t = g(X_i, \uww_i,\eta) + \psi_{it}^\top\eta\}$. By construction for any $f \in  \tilde{\mathcal{F}}$ we have the following
\begin{multline}
    \frac{1}{T} \sum_{t=1}^T \rho_{W_{it}}(W_{it} -f_t(\uww_i, X_i)) \le   \frac{1}{T} \sum_{t=1}^T (W_{it} -\psi_{it}^\top \eta)^2 \le \\
    2 +  2\eta^\top\left(\frac{1}{T}\sum_{t=1}^T\psi_{it}\psi_{it}^\top\right)\eta \le 2 + \frac{2}{T}\sum_{t=1}^T\|\psi_{it}\|_2^2C^2
\end{multline}
which allows us to use dominated convergence theorem and conclude that $R(f)$ is continuous as a function of $\eta$. Invoking Weierstrass theorem we thus conclude that there exists $\eta^{\star}$ such that the following holds:
\begin{equation}
    f(\eta^{\star}) = \argmin_{f\in \tilde{\mathcal{F}}}R(f)
\end{equation}
where $f(\eta^{\star}) \equiv g(X_i, \uww_i,\eta^{\star}) + \psi_{it}^\top\eta^{\star}$. Moment conditions then follow from the first order conditions. 
\end{proof}

For future use define $(\gamma^{\star})^{un} \equiv \xi_{it}(1-W_{it}) + (\xi_{it})_{+}W_{it}$ -- unnormalized oracle weights.

\begin{lemma}\label{lemma:emp_proc}
Suppose Assumption  \ref{as:tech_1} is satisfied. Then $\|(\weight^{\star})^{un}- \hat \weight^{un}\|_2= o_p(1)$.
\end{lemma}
\begin{proof}
We use the dual representation derived in Section \ref{sec:dual} and show that the solution converges to a population one. We use a natural adaptation of the ``small-ball'' argument from \cite{mendelson2014learning}. This is not necessary and one can construct a simpler proof using standard results for GMM estimators. We present a different argument because it can be naturally generalized to handle more sophisticated estimation procedures. For arbitrary $\eta$ define 
\begin{equation}
    f_{it}(\eta) = g(\uww_i, X_i,\eta) +\psi_{it}^{\top}\eta.\\ 
\end{equation}
Consider a lower bound on individual components of the loss function:
\begin{multline}
\rho_{W_{it}}(W_{it} -  f_{it}(\eta)) = (W_{it} -  f_{it}(\eta))^2\left(1-W_{it}\{W_{it}\le f_{it}(\eta)\}\right) = \\  (W_{it} -  f_{it}(\eta))^2\left(1-W_{it}\{W_{it} \le  f_{it}(\eta^{\star})\}\right) + \\
(W_{it} - f_{it}(\eta))^2W_{it} \left(\{W_{it} \le f_{it}(\eta^{\star})\} - \{W_{it} \le f_{it}(\eta)\}\right) \ge \\
(W_{it} - f_{it}(\eta))^2\left(1-W_{it}\{W_{it} \le f_{it}(\eta^{\star})\}\right) - \\
(W_{it} - f_{it}(\eta))^2W_{it} \{f_{it}(\eta^{\star})< 1\le f_{it}(\eta)\}.
\end{multline}
Using the properties of the oracle weights we get the following inequality for the excess loss for unit $i$:
\begin{multline}
\sum_{t=1}^T\left(\rho_{W_{it}}(W_{it} - f_{it}(\eta)) - \rho_{W_{it}}(W_{it} - f_{it}(\eta^{\star}))\right) \ge \\
\sum_{t=1}^T\left(f_{it}(\eta) - f_{it}(\eta^{\star}))^2\left(1-W_{it}\{W_{it} \le  f_{it}(\eta^{\star})\}\right)\right) +\\
\sum_{t=1}^T\left( \xi_{it}\psi_{it}^{\top}(\eta^{\star}-\eta)\left(1-W_{it}\{W_{it} \le f_{it}(\eta^{\star})\}\right)\right) -\\
\sum_{t=1}^T\left( (W_{it} - f_{it}(\eta))^2W_{it} \{f_{it}(\eta^{\star})< 1\le f_{it}(\eta)\}\right)
\end{multline}
where the last inequality follows from the fact that $\sum_{t=1}^T \xi_{it}(1-W_{it}\{W_{it} - f_{it}(\eta^{\star})\le0\})= 0$. Using this and Assumption \ref{as:tech_2} we get a lower bound:
\begin{multline}
\mathbb{P}_n\sum_{t=1}^T(1-W_{it}\{W_{it} \le f_{it}(\eta^{\star})\})(f_{it}(\eta)-f_{it}(\eta^{\star}))^2 \ge \mathbb{P}_n\sum_{t=1}^T(1-W_{it})(f_{it}(\eta)-f_{it}(\eta^{\star}))^2 \ge\\
 (\eta- \eta^{\star})^\top\left(\sum_{t=1}^T\mathbb{P}_n (1-W_{it})\Gamma_{it}\Gamma_{it}^{\top}\right)(\eta- \eta^{\star}) =  
\kappa\|\eta- \eta^{\star}\|_2^2 + o_p(\|\eta - \eta^{\star}\|_2^2).
\end{multline}
where we use Assumption \ref{as:tech_1} and the law of large numbers which holds because $\Gamma_{it}$ is bounded.

Suppose that $\|\eta - \eta^{\star}\|_2^2 = r^2$,  which by Lemma \ref{lemma:prop_ex} implies that $|g(\uww_i,X_i, \eta) -g(\uww_i,X_i, \eta^{\star})| \le C_1r$. Compactness guarantees that $\psi_{it}$ is bounded and thus $\sum_{t=1}^T\|f_t(\eta) - f_t(\eta^{\star})\|_{\infty} \le C_2r$. Using CS we get the following inequality:
\begin{multline}
 \mathbb{P}_n\xi_{it}\psi_{it}^{\top}(\eta^{\star}-\eta)\left(1-W_{it}\{W_{it}\le f_{it}(\eta^{\star})\}\right) \le \\
 \|\eta^{\star} - \eta\|_2 \times \left\|  \mathbb{P}_n\xi_{it}\psi_{it}\left(1-W_{it}\{W_{it}\le f_{it}(\eta^{\star})\}\right)\right\|_2 = \\
 \|\eta^{\star} - \eta\|_2 \times O_{p}\left(\frac{1}{\sqrt{n}}\right)
\end{multline}
We also have the following inequality:
\begin{multline}
 \mathbb{P}_n\left[\frac{1}{T}\sum_{t=1}^T(W_{it} - f_{it}(\eta))^2W_{it} \{f_{it}(\eta^{\star})< 1\le f_{it}(\eta)\}\right]\le \\
 \mathbb{P}_n\left[\frac{1}{T}\sum_{t=1}^T(f_{it}(\eta^{\star})< 1\le f_{it}(\eta))^2\{f_{it}(\eta^{\star})< 1\le f_{it}(\eta)\}\right] \le\\
 \|f^{\star} - f\|^2_{\infty}\times \mathbb{P}_n\left[\frac{1}{T}\sum_{t=1}^T\{f_{it}(\eta^{\star})< 1\le f_{it}(\eta)\}\right] 
 \end{multline}
 where the first implication follows because of the indicator, and the the second one follows by Holder inequality.
 Since $\|f^{\star} - f\|_{\infty} \le C_2 r$ we have the following:
 \begin{equation}
\mathbb{P}_n\left[\frac{1}{T}\sum_{t=1}^T\{f_{it}(\eta^{\star})< 1\le f_{it}(\eta)\}\right] \le   \mathbb{P}_n\left[\frac{1}{T}\sum_{t=1}^T\{f_{it}(\eta^{\star})< 1\le f_{it}(\eta^{\star}) + C_2 r\}\right] 
 \end{equation}
 DKW inequality implies that we have the following with high probability:
 \begin{multline}
\mathbb{P}_n\left[\frac{1}{T}\sum_{t=1}^T\{f_{it}(\eta^{\star})< 1\le f_{it}(\eta^{\star}) + C_2 r\}\right] \le \\
\mathbb{E}\left[\frac{1}{T}\sum_{t=1}^T\{f_{it}(\eta^{\star})< 1\le f_{it}(\eta^{\star}) + C_2 r\}\right] + \frac{C_3}{\sqrt{n}} = O_p\left(\frac{1}{\sqrt{n}}\right)
\end{multline}
It is now easy to see that if $r$ is greater than $O\left(\frac{1}{\sqrt{n}}\right)$ then the excess loss is positive with high probability. Since the loss function is convex this implies that optimum should belong to a ball of radius $\frac{1}{\sqrt{n}}$ around $\eta^{\star}$  with high probability which proves that for all $t$ $\|\hat\weight_{t}^{(un)} - (\weight_{t}^{\star})^{un}\|_2 =O_p\left(\frac{1}{\sqrt{n}}\right) =  o_p(1)$.
\end{proof}

\begin{lemma}
Let $R \equiv \sum_{t=1}^T\mathbb{E}(W_{it} - \mathbb{E}[W_{it}|S_i,X_i])^2$ and suppose Assumptions \ref{as:tech_1} and \ref{as:tech_2} hold. Then $R>0$ and:
\begin{equation}
\min_{\{\alpha_i\}_{i=1}^n,\eta}\sum_{t=1}^T \mathbb{P}_n \rho_{W_{it}}(W_{it} - \alpha_i - \psi_{it}^\top\eta) \ge R + o_p(1)
\end{equation}
\end{lemma}
\begin{proof}
The first statement follows from definition of $R$ and overlap assumption. The argument from the previous lemma shows that $\hat R\equiv\min_{\alpha_i,\eta}\sum_{t=1}^T \mathbb{P}_n \rho_{W_{it}}(W_{it} - \alpha_i - \psi_{it}^\top\eta)$ is bounded from below:
\begin{multline}
    \hat R \ge  \sum_{t=1}^T \mathbb{P}_n \rho_{W_{it}}(W_{it} -g(\uww_i,X_i,\eta^{\star}) - \psi_{it}^\top\eta^{\star}) + O_p\left(\frac{1}{n}\right) =\\
    \sum_{t=1}^T \mathbb{E}\left[\rho_{W_{it}}(W_{it} -g(\uww_i,X_i,\eta^{\star}) - \psi_{it}^\top\eta^{\star})\right] + o_p(1),
\end{multline}
thus we only need to bound the expectation. To this end, consider $f_t^{\star}:= \mathbb{E}[W_{it}|X_i,S_i]$, the argument below shows that this function minimizes the risk function in a large class of functions and by assumption its risk is strictly positive. Let  $\mathcal{F}:=\{ f \in L_2(\mathcal{P})^T: f_t = g(\uww_i , X_i) +h_t(S_i,X_i), g, h_t \in L_{\infty}(\mathcal{P})\}$ and observe that both $f^{\star}$ and $f(\eta^{\star})$ belong to this class. Take any function $f \in \mathcal{F}$ and consider a convex combination $f(\lambda):= f^{\star} + \lambda(f- f^{\star})$. Because $f_t\in L_{\infty}(\mathcal{P})$  and $f_{t}^{\star} \le 1- \eta$ it follows that for all $\lambda < \lambda_0$ we have $f_t(\lambda) <1$ almost surely. For any $\lambda < \lambda_0$ we have that $R(f(\lambda)) = \mathbb{E}\left[\frac{1}{T}\sum_{t=1}^T(W_{t} - f^{\star}_t)^2\right] + \mathbb{E}\left[\sum_{t=1}^T(f_t^{\star} - f_t(\lambda))^2\right]> R(f^{\star})$. The second equality here follows because $S_i$ includes $\overline W_{i}$ and thus $\frac{1}{T}\sum_{t=1}^Tf_{t}^{\star} = \overline{W}_i$. 
\end{proof}

\begin{lemma}\label{lemma:ex_sol}
Suppose that conditions of the previous Lemma hold; then with probability approaching one the primal problem  has a unique solution and the strong duality holds.
\end{lemma}
\begin{proof}
Direct application of Generalized Farkas' lemma implies that the constraint set is empty iff there exist $(\alpha^{\star}_i, \beta^{\star}_t,\gamma^{\star})$ such that the following is true:
\begin{equation}
\begin{aligned}
&\alpha_{i}^{\star} + \psi^\top_{it} \eta^{\star} \ge 0\\
&W_{it} = \{\alpha_{i}^{\star} + \psi_{it}^\top \eta^{\star} >0\}
\end{aligned}
\end{equation}
Such  $(\alpha^{\star}_i, \eta^{\star})$ exist iff $\min_{\alpha_i,\eta}\sum_{t=1}^T \mathbb{P}_n \rho_{W_{it}}(W_{it} - \alpha_i - \psi_{it}^\top\eta) = 0$ which by previous lemma does not happen with probability approaching one. Thus the constraint set is not empty and convex. Since the objective function is strictly convex we have that the primal problem has the unique solution. Since all the inequality constrains are affine strong duality holds (see 5.2.3 in \cite{boyd2004convex}) and we have the result.
\end{proof}

\subsection{Theorems}
\textbf{Proof of Theorem \ref{th:main_inf_theorem}}: Existence of $\hat \gamma_{it}$ follows from \ref{lemma:ex_sol}. The estimator has the following representation:
\begin{multline}
\hat \tau - \tau_{cond} =  \frac{1}{NT}\sum_{it}\hat \gamma_{it} \epsilon_{it}=\\\frac{1}{\mathbb{P}_n \frac{1}{T} \sum_{t=1}^T \hat \weight_{it}^{un}  W_{it}}
 \left(\frac{1}{nT} \sum_{it} (\weight^{\star}_{it})^{un}\epsilon_{it} +  \frac{1}{nT} \sum_{it} (\hat \weight_{it}^{un} -(\weight_{it}^{\star})^{un})\epsilon_{it}\right)
\end{multline}
By assumption we have the following:
\begin{multline}
\mathbb{E}[ (\hat \weight_{it}^{un} -(\weight_{it}^{\star})^{un})\epsilon_{it}|\{\uww_j, X_j\}_{j=1}^n] = (\hat \weight_{it}^{un} -(\weight_{it}^{\star})^{un})\mathbb{E}[\epsilon_{it}|\{\uww_j, X_j\}_{j=1}^n] = \\
 (\hat \weight_{it}^{un} -(\weight_{it}^{\star})^{un})\mathbb{E}[\epsilon_{it}|\uww_i, X_i] = 0
\end{multline}
This implies that by conditional Chebyshev inequality we have the following:
\begin{multline}
\zeta_n(\epsilon)\equiv\mathbb{E}\left[\left\{\sqrt{n}\left|\mathbb{P}_n\frac{1}{T}\sum_{t=1}^T (\hat \weight_{it}^{un} -(\weight_{it}^{\star})^{un})\epsilon_{it}\right|\ge \delta\right\}|\{\uww_j, X_j\}_{j=1}^n\right] \le\\
 \frac{\mathbb{P}_n\mathbb{E}\left[\left(\sum_{t=1}^T( \hat \weight_{it}^{un} -(\weight_{it}^{\star})^{un})\epsilon_{it}^2\right)^2|\{\uww_j, X_j\}_{j=1}^n \right]}{T^2\delta^2}\le\frac{\overline{\sigma}^2_u}{T\epsilon^2} \|(\weight^{\star})^{un}- \hat \weight^{un}\|_2^2= o_p(1)
\end{multline}
Since indicator is a bounded function it follows that for any $\delta>o$
\begin{equation}
\mathbb{E}[\zeta_n(\epsilon)] = o(1)
\end{equation}
and thus we have $\frac{1}{nT} \sum_{it}((\weight_{it}^{\star})^{un}- \hat \weight_{it}^{un})\epsilon_{it} = o_p\left(\frac{1}{\sqrt{n}}\right)$. Finally we need to check that CLT applies to $  \frac{1}{nT} \sum_{it} (\weight^{\star}_{it})^{un}\epsilon_{it}$. The mean of each summand is zero and the variance is bounded:
\begin{equation}
\mathbb{E}\left[\left( \frac{1}{T} \sum_{t=1}^T (\weight^{\star}_{it})^{un}\epsilon_{it}\right)^2\right] \le \frac{1}{T} \sum_{t=1}^T \mathbb{E}\left[\left( (\weight^{\star}_{it})^{un}\epsilon_{it}\right)^2\right] \le  \sum_{t=1}^T \sqrt{\mathbb{E}[\epsilon_{it}^4]\mathbb{E}[( (\weight^{\star}_{it})^{un})^4]} <\infty
\end{equation}
Finally, define:
\begin{equation}
\weight_{it}^{\star}\equiv \frac{(\weight_{it}^{\star})^{un}}{ \mathbb{E}\left[ \frac{1}{T} \sum_{t=1}^T  (\weight_{it}^{\star})^{un} W_{it} \right] }
\end{equation}
From Lemma \ref{lemma:emp_proc} we get
\begin{equation}
\mathbb{P}_n \frac{1}{T} \sum_{t=1}^T \hat \weight_{it}^{un} W_{it} = \mathbb{E}\left[ \frac{1}{T} \sum_{t=1}^T  (\weight_{it}^{\star})^{un} W_{it} \right] +o_p(1)
\end{equation}
and the result follows.

\textbf{Proof of Theorem \ref{th:bootstrap}}:
For each observation $i$ define $M_i$ -- the number of times this observation is sampled in a bootstrap sample. Using this notation we can define bootstrap analogs of $\alpha_i$ and $\gamma$ from the proof of Lemma \ref{lemma:emp_proc}:
\begin{equation}
 \eta^{(b)} = \arg\min\mathbb{P}_n M_i \frac{1}{T} \sum_{t=1}^T\rho_{W_{it}}(W_{it} -g(\uww_i, X_i, \eta)- \psi_{it}^T\eta).
\end{equation}
 It is straightforward to extend the proof of Lemma \ref{lemma:emp_proc} and show that bootstrap weights converge to population ones. Most part follow because of  two key properties of $\{M_i\}_{i=1}^n$:
\begin{equation}
\begin{aligned}
&\mathbb{P}_n M_i X_i = \mathbb{E}[X_i] + o_p(1)\\
&\mathbb{P}_n M_i \epsilon_i = O_p\left(\frac{1}{\sqrt{n}}\right)
\end{aligned}
\end{equation}
for any square integrable $X_i$ and any square integrable mean-zero $\epsilon_i$ (all independent of $M_i$). The second inequality follows by applying Chebyshev inequality, the first one follows from the second one. The only additional result that we need is the following one:
\begin{multline}
\mathbb{P}_nM_i\left[\frac{1}{T}\sum_{t=1}^T\{f_{it}^{\star}< 1\le f_{it}^{\star} + C_2 r\}\right] = \mathbb{P}_n(M_i-1)\left[\frac{1}{T}\sum_{t=1}^T\{f_{it}^{\star}< 1\le f_{it}^{\star} + C_2 r\}\right] + \\
 \mathbb{P}_n\left[\frac{1}{T}\sum_{t=1}^T\{f_{it}^{\star}< 1\le f_{it}^{\star} + C_2 r\} - \mathbb{E}\left[\frac{1}{T}\sum_{t=1}^T\{f_{it}^{\star}< 1\le f_{it}^{\star} + C_2 r\}\right] \right] + \\
 \mathbb{E}\left[\frac{1}{T}\sum_{t=1}^T\{f_{it}^{\star}< 1\le f_{it}^{\star} + C_2 r\}\right] =  \mathbb{E}\left[\frac{1}{T}\sum_{t=1}^T\{f_{it}^{\star}< 1\le f_{it}^{\star} + C_2 r\}\right]  + O_p\left(\frac{1}{\sqrt{n}}\right)
\end{multline}
where the last line follows by DKW inequality, the fact that the set of intervals is Donsker, and the multiplier process converges to same limit process as the standard empirical one. 
It follows that we have convergence results:
\begin{equation}
\begin{aligned}
&\|\weight^{(b)} - \weight^{\star}\|_2 = O_p\left(\frac{1}{\sqrt{n}}\right).
\end{aligned}
\end{equation}
By construction of bootstrap estimator we have the following representation:
\begin{multline}
\hat \tau^{(b)} - \hat \tau = \mathbb{P}_n\frac{1}{T}\sum_{t=1}^T\left(M_i\weight_{it}^{(b)} - \hat\weight_{it}\right)W_{it}\tau_t(\uww_i, X_i) +
 \mathbb{P}_n \frac{1}{T}\sum_{t=1}^T(M_i \weight_{it}^{(b)} -\hat\weight_{it}) \epsilon_{it} = \\
 \mathbb{P}_n\frac{1}{T}\sum_{t=1}^T\left(M_i\weight_{it}^{(b)} - \hat\weight_{it}\right)W_{it}\tau_t(\uww_i, X_i) +
 \mathbb{P}_n \frac{1}{T}\sum_{t=1}^T \gamma_{it}^{\star} \epsilon_{it} + o_p\left(\frac{1}{\sqrt{n}}\right)
 \end{multline}
From this representation it follows that if $\tau_{it} = const$ then the bootstrap estimator is consistent for the asymptotic variance of $\hat \tau$. Otherwise, since the second summand is uncorrelated with the first one we have that the bootstrap variance is a conservative estimator of the correct variance.

\end{appendix}

\end{document}